\documentclass[11pt]{article} 
\usepackage{a4wide} 
\usepackage[T1]{fontenc} 
\usepackage{lmodern}
\usepackage{amsmath} 
\usepackage{amsthm} 
\usepackage{amssymb}
\usepackage{stmaryrd} 
\let\cal\mathcal  

\newtheorem{thm}{Theorem} \newtheorem{lem}[thm]{Lemma}
\newtheorem{clm}[thm]{Claim} \newtheorem{cor}[thm]{Corollary}
\newtheorem{probl}{Open Problem}

\theoremstyle{remark} \newtheorem{remark}{Remark}
\newtheorem{dfn}{Definition} 
\newtheorem{obs}{Observation}

\newcommand{\Cite}[2]{{#1~}\cite{#2}} 
\let\geq\geqslant
\let\leq\leqslant 
 
\newcommand{\gf}[1]{GF(#1)} 
\renewcommand{\deg}[1]{d(#1)}
\newcommand{\length}[1]{l(#1)}

\newcommand{\dg}[1]{d_{#1}} 
\newcommand{\size}[1]{\mathrm{Size}(#1)}
\newcommand{\henv}[1]{ {#1}_{\mathrm{he}} } 
\newcommand{\lenv}[1]{{#1}_{\mathrm{le}} } 
\newcommand{\lmin}[1]{ {#1}_{\mathrm{min}} }
\newcommand{\lmax}[1]{ {#1}_{\mathrm{max}} }
\newcommand{\overbar}[1]{\mkern
  1.5mu\overline{\mkern-1.5mu#1\mkern-1.5mu}\mkern 1.5mu}
\renewcommand{\underbar}[1]{\mkern
  1.5mu\underline{\mkern-1.5mu#1\mkern-1.5mu}\mkern 1.5mu}
\newcommand{\lsat}[1]{ {\underbar{#1}} } 
\newcommand{\hsat}[1]{{\overbar{#1}} } 
\newcommand{\Lmin}[1]{ {\pol{#1}}_{\mathrm{min}} }
\newcommand{\Lmax}[1]{ {\pol{#1}}_{\mathrm{max}} }
\newcommand{\pr}{\times} \newcommand{\su}{+}
\newcommand{\nulis}{\mathsf{0}} 
\newcommand{\vienas}{\mathsf{1}}
\newcommand{\eq}{\rightleftharpoons} 
\newcommand{\A}{\mathbf{S}}
\newcommand{\B}{\mathbf{B}} 
\newcommand{\R}{\mathbf{A}}
\newcommand{\Mult}[2]{{#1}_{\mathrm{lin}}(#2) }
\renewcommand{\mathbf}[1]{\mathrm{#1}} 
\newcommand{\M}{\mathbf{Min}^-}
 
\newcommand{\Min}{\mathbf{Min}}
\newcommand{\Max}{\mathbf{Max}} 
\newcommand{\MAX}{\mathbf{Max}^-}
\newcommand{\RR}{\mathbb{R}} 
\newcommand{\NN}{\mathbb{N}}
\newcommand{\ZZ}{\mathbb{Z}} 
\newcommand{\f}{\mathcal{F}}
\newcommand{\F}{\mathsf{F}} 
\newcommand{\G}{\mathsf{G}}
\renewcommand{\H}{\mathsf{H}} 

\newcommand{\bound}[1]{\mu(#1)}
\newcommand{\bbound}[1]{\mathrm{sep}(#1)}
\newcommand{\pol}[1]{\hat{#1}}
\newcommand{\PERM}[1]{\mathrm{PER}_{#1}}

\newcommand{\HC}[1]{\mathrm{HC}_{#1}}
\newcommand{\STCONN}[1]{\mathrm{STCON}_{#1}}
\newcommand{\SSTCONN}[1]{f_{#1}}

\newcommand{\CONN}[1]{\mathrm{CONN}_{#1}}
\newcommand{\MP}[1]{\mathrm{MP}_{#1}}
\newcommand{\APSP}[1]{\mathrm{APSP}_{#1}}
\newcommand{\ST}[1]{\mathrm{ST}_{#1}}
\newcommand{\Clique}[1]{\mathrm{CLIQUE}_{#1}}
\newcommand{\tR}[1]{\mathrm{TR}_{#1}}
\newcommand{\tRP}[1]{\mathrm{TR}^{\ast}_{#1}}
\newcommand{\Prod}[1]{\mathrm{pol}(#1)}
\newcommand{\ext}[1]{\mathrm{ext}(#1)}
\newcommand{\Ext}[2]{\mathrm{ext}_{#2}(#1)}
\newcommand{\ddeg}[2]{\#_{#2}(#1)}
 
\newcommand{\gate}{u}
 
\newcommand{\match}[1]{m(#1)}
\newcommand{\up}[1]{\overline{#1}} 
\newcommand{\ff}[3]{#1_{#2}^{[#3]}}

\newcommand{\factor}[1]{K_{#1}}
\newcommand{\Factor}[2]{K_{#1}(#2)} 

\newcommand{\weight}[1]{\mu(#1)} 
\newcommand{\wweight}{\mu}

\newcommand{\depth}[1]{\mathrm{Depth}[#1]}

\newcommand{\Der}[2]{\partial #1/\partial #2}

\begin{document}

\title{Lower Bounds for Tropical Circuits and Dynamic
  Programs\thanks{Research supported by the DFG grant SCHN~503/6-1.}}

\author{Stasys Jukna\thanks{University of Frankfurt, Institute of
    Computer Science, D-60054 Frankfurt, Germany. Affiliated with
    Vilnius University, Institute of Mathematics and Informatics,
    Vilnius, Lithuania. Email: jukna@thi.informatik.uni-frankfurt.de}}

\maketitle

\begin{abstract}
  Tropical circuits are circuits with Min and Plus, or Max and Plus
  operations as gates.  Their importance stems from their intimate
  relation to dynamic programming algorithms.  The power of tropical
  circuits lies somewhere between that of monotone boolean circuits
  and monotone
  arithmetic circuits. In this paper we survey known and present some new  lower bounds arguments for tropical circuits, and hence, for dynamic programs.\\[2ex]
  {\bf Keywords:} Tropical circuits, dynamic programming, monotone
  arithmetic circuits, lower bounds.
\end{abstract}

\section{Introduction}

Understanding the power and limitations of fundamental algorithmic
pa\-ra\-digms---such as greedy or dynamic programming---is one of the
basic questions in the algorithm design and in the whole theory of
computational complexity. In this paper we focus on the dynamic
programming paradigm.

Our starting point is a simple observation that many dynamic
programming algorithms for optimization problems are just recursively
constructed \emph{circuits} over the corresponding semirings. Each
such circuit computes, in a natural way, some polynomial over the
underlying semiring. Most of known dynamic programming algorithms
correspond to circuits over the $(\min,+)$ or $(\max,+)$ semirings,
that is, to \emph{tropical circuits}.\footnote{There is nothing
  special about the term ``tropical''.  Simply, this term is used in
  honor of Imre Simon who lived in Sao Paulo (south tropic).  Tropical
  algebra and tropical geometry are now intensively studied topics in
  mathematics. } Thus, lower bounds for tropical circuits show the
limitations of dynamic programming algorithms over the corresponding
semirings.

The power of tropical circuits (and hence, of dynamic programming)
lies somewhere between that of monotone boolean circuits and monotone
arithmetic circuits:
\[
\mbox{monotone boolean \ $\leq$\ tropical\ $\leq$\ monotone
  arithmetic}
\]
and the gaps may be even exponential (we will show this in
Section~\ref{sec:gaps}).

Monotone \emph{boolean} circuits are most powerful among these three
models and, for a long time, only linear lower bounds were known for
such circuits.  First super-polynomial lower bounds for the $k$-clique
function $\Clique{}$ and the perfect matching function $\PERM{}$ were
proved by \Cite{Razborov}{razb1,razb2} by inventing his method of
approximations. At almost about the same time, explicit exponential
lower bounds were also proved by
\Cite{Andreev}{andr1,andr2}. \Cite{Alon and Boppana}{AB87} improved
Razborov's lower bound for $\Clique{}$ from super-polynomial until
exponential. Finally, \Cite{Jukna}{Juk99} gave a general and easy to
apply lower bounds criterium for monotone boolean and real-valued
circuits, yielding strong lower bounds for a row of explicit boolean
functions. These lower bounds hold for tropical circuits as well.

On the other hand, monotone \emph{arithmetic} circuits are much easier
to analyze: such a circuit cannot produce anything else but the
monomials of the computed polynomial, no ``simplifications'' (as
$x^2=x$ or $x+xy=x$) are allowed here. Exponential lower bounds on the
monotone arithmetic circuit complexity were proved already by
\Cite{Schnorr}{schnorr} (for $\Clique{}$), and \Cite{Jerrum and
  Snir}{jerrum} (for $\PERM{}$ and some other polynomials). A
comprehensive survey on arithmetic (not necessarily monotone) circuits
can be found in the book by \Cite{Shpilka and Yehudayoff}{shpilka}.

In this paper we summarize our knowledge about the power of tropical
circuits.  As far as we know, no similar attempt was undertaken in
this direction after the classical paper by \Cite{Jerrum and
  Snir}{jerrum}.  The main message of the paper is that not only
methods developed for monotone \emph{boolean} circuits, but
(sometimes) even those for a much weaker model of monotone
\emph{arithmetic} circuits can be used to establish limitations of
dynamic programming. Although organized as a survey, the paper
contains some new results, including:
\begin{enumerate}
\item A short and direct proof that tropical circuits for optimization
  problems with \emph{homogeneous} target polynomials are not more
  powerful than monotone arithmetic circuits
  (Theorem~\ref{thm:homog}).  This explains why we do not have
  efficient dynamic programming algorithms for optimization problems
  whose target sums all have the same length.  In the case of
  $\Min$-semirings, this was proved by \Cite{Jerrum and Snir}{jerrum}
  using the Farkas lemma.

\item A new and simple proof of \Cite{Schnorr's}{schnorr} lower bound
  on the size of monotone arithmetic circuits computing so-called
  ``separated'' polynomials (Theorem~\ref{thm:schnorr}).  A polynomial
  $f$ is \emph{separated} if the product of any two of its monomials
  contains no third monomial of $f$ distinct from these two ones.

\item A new and simpler proof of \Cite{Gashkov and
    Sergeev's}{gashkov,GS} lower bound on the size of monotone
  arithmetic circuits computing so-called ``$k$-free'' polynomials
  (Theorem~\ref{thm:GS}).  A polynomial is $k$-\emph{free} if it does
  not contain a product of two polynomials, both with more than $k$
  monomials. This extend's Schnorr's bound, since every separated
  polynomial is also $1$-free.

\item An easy to apply ``rectangle'' lower bound
  (Lemma~\ref{lem:simple}).

\item A truly exponential lower bound for monotone arithmetic circuits
  using expander graphs (Theorem~\ref{thm:trully}).

\end{enumerate}

\section{Semirings}
\label{sec:semirings}

A (commutative) semiring is a system $\A=(S,\su,\pr,\nulis,\vienas)$,
where $S$ is a set, $\su$ (``sum'') and $\pr$ (``product'') are binary
operations on $S$, and $0$ and $1$ are elements of $S$ having the
following three properties:
\begin{description}
\item[(i)] in both $(S,\su,\nulis)$ and $(S,\pr,\vienas)$, operation
  are associative and commutative with identities $\nulis$ and
  $\vienas$: $a\su \nulis=a$ and $a\pr \vienas=a$ hold for all $a\in
  S$;
\item[(ii)] product distributes over sum: $a\pr(b\su c)=(a\pr
  b)\su(a\pr c)$;
\item[(iii)] $a\pr \nulis=\nulis$ for all $a\in S$ (``annihilation''
  axiom).
\end{description}
A semiring is \emph{additively-idempotent} if $a\su a=a$ holds for all
$a\in S$, and is \emph{mul\-ti\-pli\-ca\-ti\-ve\-ly-idempotent} if
$a\pr a=a$ holds for all $a\in S$.

We will use the common conventions to save parenthesis by writing
$a\pr b\su c\pr d$ instead of $(a\pr b)\su(a\pr c)$, and replacing
$a\pr b$ by $ab$. Also, $a^n$ will stand for $a\pr a\pr \cdots \pr a$
$n$-times. If desired, we will also assume that the sets $\NN$, $\ZZ$
or $\RR$ also contain $+\infty$ and/or $-\infty$.

In this paper, we will be interested in the following semirings:
\begin{itemize}
\item Arithmetic semiring $\R=(\NN, +, \cdot,0,1)$.
\item Boolean semiring $\B=(\{0,1\},\lor,\land,0,1)$.
\item Min-semirings $\Min =(\NN,\min,+,+\infty,0)$ and
  $\M=(\ZZ,\min,+,+\infty,0)$.
\item Max-semirings $\Max=(\NN, \max, +,-\infty,0)$ and $\MAX=(\ZZ,
  \max, +,-\infty,0)$.
\item Min- and Max-semirings are called \emph{tropical} semirings.
\end{itemize}
Note that all these semirings, but $\R$, are additively-idempotent,
and none of them, but $\B$, is multiplicatively-idempotent.  Note also
that in arithmetic and in tropical semirings one usually allows
rational or even real numbers, not just integers. This corresponds to
considering optimization problems with real, not necessarily integral
``weights''. The point, however, is that lower-bound techniques, we
will consider below, work already on smaller domains.  In fact, they
work when, besides $\infty$ or $-\infty$, the domain contains $0$ and
$1$ or $0$ and $-1$. Roughly speaking, the larger is the domain, the
easier is to prove lower bounds over them. In particular, the bounds
remain true in larger domains as well.

Due to their intimate relation to discrete optimization, we will be
mainly interested in tropical semirings, and circuits over these
semirings. Lower bounds for such circuits give lower bounds for the
number of subproblems used by dynamic programming algorithm.  The
semirings $\M$ and $\MAX$ are isomorphic via the transformation
$x\mapsto -x$, so we will not consider $\MAX$ separately: all results
holding for $\M$ hold also for $\MAX$.

\section{Polynomials}

Let $\A=(S,\su,\pr,\nulis,\vienas)$ be a semiring, and let
$x_1,\ldots,x_n$ be variables ranging over~$S$.  A \emph{monomial} is
any product of these variables, where repetitions are allowed. By
commutativity and associativity, we can sort the products and write
monomials in the usual notation, with the variables raised to
exponents.  Thus, every monomial $x_1^{a_1}x_2^{a_2}\cdots x_n^{a_n}$
is uniquely determined by the vector of exponents
$(a_1,\ldots,a_n)\in\NN^n$, where $x_i^{0}=\vienas$.  (Note that in
tropical semirings, monomials are linear combinations
$a_1x_1+a_2x_2+\cdots+a_n x_n$, that is, sums, not products.)  The
\emph{degree}, $|p|$, of a monomial is the sum $|p|=a_1+\cdots+a_n$ of
its exponents. A monomial $p$ is \emph{multilinear} if every exponent
$a_i$ is either $0$ or $1$.  A monomial $p=x_1^{a_1}\cdots x_n^{a_n}$
\emph{contains} a monomial $q=x_1^{b_1}\cdots x_n^{b_n}$ (or $q$ is a
\emph{factor} of $p$) if $a_i\geq b_i$ for all $i=1,\ldots,n$, that
is, if $p=qq'$ for some monomial~$q'$.

By a \emph{polynomial}\footnote{Usually, polynomials of more than one
  variable are called \emph{multivariate}, but we will omit this for
  shortness.}  we will mean a finite sum of monomials, where
repetitions of monomials are allowed.  That is, we only consider
polynomials with nonnegative integer coefficients.  A polynomial is
\emph{homogeneous} if all its monomials have the same degree, and is
\emph{multilinear} if all its monomials are multilinear (no variables
of degree $>1$). For example, $f=x^2y \su xyz$ is homogeneous but not
multilinear, whereas $g= x\su yz$ is multilinear but not homogeneous.
The sum and product of two polynomials is defined in the standard way.
For polynomials $f,h$ and a monomial $p$, we will write:
\begin{itemize}

\item $f=h$ if $f$ and $h$ have the same monomials appearing not
  necessarily with the same coefficients;
\item $f\eq h$ if $f$ and $h$ have the same monomials appearing with
  the same coefficients;

\item $f\subseteq h$ if every monomial of $f$ is also a monomial of
  $h$;

\item $p\in f$ if $p$ is a monomial of $f$;

\item $|f|$ to denote the number of \emph{distinct} monomials in $f$;

\item $X_p$ to denote the set of variables appearing in $p$ with
  non-zero degree.

\end{itemize}

Every polynomial $f(x_1,\ldots,x_n)$ defines a function
$\pol{f}:S^n\to S$, whose value $\pol{f}(s_1,\ldots,s_n)$ is obtained
by substituting elements $s_i\in S$ for $x_i$ in $f$. Polynomials $f$
and $g$ are \emph{equivalent} (or \emph{represent the same function})
over a given semiring, if $\pol{f}(s)=\pol{h}(s)$ holds for all $s\in
S^n$.  It is important to note that the same polynomial $f(x)\eq
\sum_{I\in{\cal I}} c_I\prod_{i\in I}x_i^{a_i}$ represents different
functions over different semirings:
\begin{xalignat*}{2}
  \pol{f}(x)&= \sum_{I\in{\cal I}} c_I\prod_{i\in I}x_i^{a_i}
  &&\mbox{over $\R$
    (counting)}\\
  \pol{f}(x)&= \bigvee_{I\in{\cal I}}\ \bigwedge_{i\in I}x_i &&\mbox{over $\B$ (existence)}\\
  \pol{f}(x)&= \min_{I\in{\cal I}}\ \sum_{i\in I}a_ix_i  &&\mbox{over $\Min$ and $\M$ (minimization)}\\
  \pol{f}(x)&= \max_{I\in{\cal I}}\ \sum_{i\in I}a_ix_i &&\mbox{over
    $\Max$ and $\MAX$ (maximization)}
\end{xalignat*}
Note that in the boolean semiring as well as in all four tropical
semirings, the coefficients $c_I$ do not influence the computed value
$\pol{f}(x)$, and we can assume that $c_I=1$ for all $I\in{\cal I}$;
this is because, say, $\min\{x,x,y\}=\min\{x,y\}$.  The degrees,
however, are important: say, $\min\{2x,y\}\neq \min\{x,y\}$.

\section{Structure of Equivalent Polynomials}
Let $f$ and $h$ be any two polynomials on the same set of variables.
In general, if $f$ and $h$ are equivalent (i.e. if $\pol{f}=\pol{h}$
holds) over some semiring, then neither $f\eq h$ nor even $f=h$ need
to hold.  The arithmetic semiring is here an exception.

\begin{lem}\label{lem:arithm0}
  If $\pol{f}=\pol{h}$ holds over the arithmetic semiring $\R$, then
  $f\eq h$.
\end{lem}

\begin{proof}
  There are several ways to prove this fact. We follow the argument
  suggested by Sergey Gashkov (personal communication).  Suppose that
  $\pol{f}=\pol{h}$ but $f\not\eq h$. Since $f\not\eq h$, the
  polynomial $g=f-h$ contains at least one monomial. Let $p$ be a
  monomial of $g$ of maximum degree.  Take all partial derivatives of
  $g$ with respect to the variables of $p$ until all they
  disappear. Since $p$ has maximum degree, we obtain some constant
  $\neq 0$.  But since $\pol{g}=\pol{f}-\pol{h}$ is a zero function,
  the derivative should be zero, a contradiction.
\end{proof}

In tropical semirings, we only have weaker structural properties.  For
a polynomial $f$, let $\lmin{f}\subseteq f$ denote the set of all
monomials of $f$ not containing any other monomial of $f$, and
$\lmax{f}\subseteq f$ denote the set of all monomials of $f$ not
contained in any other monomial of~$f$.  For example, if $f=\{x, x^2y,
yz\}$, then $\lmin{f}=\{x,yz\}$ and $\lmax{f}=\{x^2y,yz\}$.  Note that
every monomial of $f$ contains (properly or not) at least one monomial
of $\lmin{f}$, and is contained in at least one monomial of
$\lmax{f}$.  Note also that $\Lmin{f}=\pol{f}$ holds in $\Min$
semirings, and $\Lmax{f}=\pol{f}$ holds in $\Max$ semirings.

\begin{lem}\label{lem:arithm1}
  If $\pol{f}=\pol{h}$ holds over $\Min$, and if $h$ is multilinear,
  then $\lmin{f}=\lmin{h}$.
\end{lem}

\begin{proof}
  Let us first show that every monomial of $f$ must contain at least
  one monomial of $h$, and hence, of $\lmin{h}$.  To see this, assume
  that there is a monomial $p\in f$ which contains no monomial of
  $h$. Since $h$ is multilinear, this means that every monomial of $h$
  must contain a variable not in $X_p$.  So, on the assignment $a_p$
  which sets to $1$ all variables in $X_p$, and sets to $\infty$ all
  the remaining variables, we have that $\pol{h}(a_p)=\infty$.  But
  $\pol{f}(a_p)\leq \pol{p}(a_p)=|X_p|<\infty$, a contradiction with
  $\pol{f}=\pol{h}$.

  Since no monomial in $\lmin{f}$ can contain another monomial of $f$,
  it remains therefore to show that $\lmin{h}\subseteq f$. For this,
  assume that there is a monomial $q\in \lmin{h}$ such that $q\not\in
  f$. If we take the assignment $a_q$, then $\pol{h}(a_q)
  =\pol{q}(a_q)=|X_q|$. On the other hand, the assignment $a_q$ sets
  all monomials $p\in f$ such that $X_p\not\subseteq X_q$ to
  $\infty$. Each of the remaining monomials $p\in f$ (if there is any)
  must satisfy $X_p\subseteq X_q$. But we already know that $p$ must
  contain some monomial $q'\in \lmin{h}$, that is, $X_{q'}\subseteq
  X_p\subseteq X_q$. Since both monomials $q$ and $q'$ are multilinear
  and belong to $\lmin{h}$, this implies $q=q'$, and hence, also
  $X_p=X_q$. Since $q$ is multilinear and $p\neq q$, this means that
  $p$ must have strictly larger degree $|p|$ than $|X_q|$, and hence,
  $\pol{p}(a_q)=|p|>|X_q|=\pol{h}(a_q)$, a contradiction with
  $\pol{f}=\pol{h}$.
\end{proof}

\begin{remark}
  Note that Lemma~\ref{lem:arithm1} needs not to hold, if both
  polynomials are not multilinear. Say, if $f=\min\{2x, x+y, 2y\}$ and
  $h=\min\{2x, 2y\}$, then $\pol{f}=\pol{h}$ holds (because $x+y\geq
  \min\{2x,2y\}$), but $\lmin{f}=f\neq h=\lmin{h}$.  In this example,
  monomial $x+y=\tfrac{1}{2}(2x) +\tfrac{1}{2}(2y)$ is a convex
  combination of the monomials $2x$ and $2y$. And in fact, using the
  Farkas lemma about solvability of systems of linear inequalities,
  \Cite{Jerrum and Snir}{jerrum} have proved that, if $f$ and $h$ are
  arbitrary (not necessarily multilinear) polynomials such that
  $\pol{f}=\pol{h}$ holds over $\Min$, then there is a set
  $h'\subseteq h$ of monomials such that $h'\subseteq f$, and every
  monomial of $f\cup h$ is at least some convex combination of the
  monomials in $h'$.
\end{remark}

\begin{lem}\label{lem:arithm2}
  If $\pol{f}=\pol{h}$ holds over $\Max$, and if $h$ is multilinear,
  then $f$ is also multilinear, and $\lmax{f}=\lmax{h}$.
\end{lem}

\begin{proof}
  Assume that $f$ is not multilinear.  Then $f$ contains a monomial
  $p$ (sum) in which some variable $x_i$ appears more than once. If we
  set this variable to $1$ and the rest to $0$, then $\pol{h}$ takes
  some value $\leq 1$, but $\pol{f}$ takes value $|p|\geq 2$, a
  contradiction with $\pol{f}=\pol{h}$.  Thus, both polynomials $f$
  and $h$ must be multilinear.

  We claim that every monomial of $f$ must be contained in at least
  one monomial of $h$.  Indeed, if some monomial $p\in f$ is contained
  in none of the monomials $q\in h$, then every monomial $q\in h$ is
  missing at least one variable from $X_p$.  So, on the assignment
  $b_p$ which sets to $1$ all variables in $X_p$, and sets to $0$ all
  the remaining variables, we have that $\pol{h}(b_p)\leq
  |X_p|-1$. But $\pol{f}(b_p)\geq \pol{p}(b_p)=|X_p|$, a contradiction
  with $\pol{f}=\pol{h}$.

  It remains therefore to show that $\lmax{h}\subseteq f$. For this,
  assume that there is a monomial $q\in \lmax{h}$ such that $q\not\in
  f$. If we take the assignment $b_q$, then $\pol{h}(a_q)
  =\pol{q}(a_q)=|X_q|$. On the other hand, $X_p\not\supseteq X_q$ must
  holds for every monomial $p\in f$, implying that $\pol{f}(b_q)\leq
  |X_q|-1$.  Indeed, we already know that every monomial $p\in f$ must
  be contained in some monomial $q'\in \lmax{h}$. So, $X_p\supseteq
  X_q$ implies $X_{q'}\supseteq X_p\supseteq X_q$. Since both
  monomials $q$ and $q'$ belong to $\lmax{h}$, this implies $X_p=X_q$,
  and hence also $p=q$ because both monomials $p$ and $q$ are
  multilinear. This contradicts our assumption $q\not\in f$.
\end{proof}

In tropical semirings $\M$ and $\MAX$, we have an even stronger
property.

\begin{lem}\label{lem:multilin}
  If $\pol{f}=\pol{h}$ holds over a tropical semiring $\M$ or $\MAX$,
  and if $h$ is multilinear, and then $f= h$.
\end{lem}

\begin{proof}
  We claim that the polynomial $f$ must be also multilinear. To see
  this, assume that $f$ contains a monomial $p$ (sum) in which some
  variable $x_i$ appears more than once. Then, in the semiring $\M$,
  we can set $x_i=-1$ and $x_j=0$ for all $j\neq i$. Under this
  assignment, we have $\pol{h}(x)\geq -1$, because all monomials of
  $h$ get value $\geq -1$, but $\pol{f}(x)\leq -2$ since already the
  monomial $p$ of $f$ gets value $\leq -2$, a contradiction. The the
  $\MAX$ semiring (and even in $\Max$), it is enough to set $x_i=1$
  and $x_j=0$ for all $j\neq i$ to get the desired contradiction.

  Let us now show that $f=h$ must hold over the semiring $\M$; the
  argument for $\MAX$ is similar. We know that both polynomials $f$
  and $h$ are multilinear.  Hence, Lemma~\ref{lem:arithm1} implies
  that $\lmin{f}=\lmin{h}$ (this holds even in $\Min$). In particular,
  every monomial of $f$ must contain at least one monomial of $h$, and
  every monomial of $h$ must contain at least one monomial of $f$.
  Thus, $h\not\subseteq f$ can only happen, if there is a monomial
  $p\in h$ such that, for every monomial $q\in f$, we have that either
  $X_p\not\supseteq X_q$ or $X_p\supset X_q$ (proper inclusion). In
  any case, every monomial $q\in f$ misses some variable of $p$.  So,
  if we assign $-1$ to all variables of $p$, and $0$ to the remaining
  variables, then $h$ takes some value $\leq -|X_p|$. But since each
  monomial $q\in f$ misses at least one variable of $p$, the value of
  each of these monomials, and hence the value of $f$, must be $\geq
  |X_p|+1$, a contradiction with $\pol{f}=\pol{h}$.  This shows
  $h\subseteq f$. The proof of the converse inclusion $f\subseteq h$
  is the same.
\end{proof}

Note that for non-multilinear polynomials $f$,
Lemma~\ref{lem:multilin} needs not to hold. For example, If
$f=\min\{x,2x,3x\}$ and $h=\min\{x,3x\}$, then $\pol{f}=\pol{h}$ holds
over $\M$, but $f\neq h$.

\section{Circuits and their Polynomials}
\label{sec:circuits}

A \emph{circuit} $\F$ over a semiring $\A=(S,\su,\pr,\nulis,\vienas)$
is a usual fanin-$2$ circuit whose inputs are variables
$x_1,\ldots,x_n$ and constants $\nulis$ and $\vienas$. Gates are
fanin-$2$ $\su$ and $\pr$. That is, we have a directed acyclic graph
with $n+2$ fanin-$0$ nodes labeled by
$x_1,\ldots,x_n,\nulis,\vienas$. At every other node, the sum $(\su$)
or the product ($\pr$) of its entering nodes is computed; nodes with
assigned operations are called \emph{gates}.  The \emph{size} of $\F$,
denoted by $\size{\F}$, is the number of gates in $\F$. The
\emph{depth} is the largest number of edges in a path from an input
gate to an output gate.

Like polynomials, circuits are also ``syntactic'' objects. So, we can
associate with every circuit $\F$ the unique polynomial $F$
\emph{produced} by $\F$ inductively as follows:\footnote{We will
  always denote circuits as upright letters $\F, \G, \H, \ldots$, and
  their produced polynomials by italic versions $F,G, H,\ldots$.}
\begin{itemize}
\item If $\F=x_i$, then $F\eq x_i$.

\item If $\F=\G\su \H$, then $F\eq \sum_{p\in G} p \su \sum_{q\in H}
  q$.

\item If $\F=\G\pr \H$, then $F\eq \sum_{p\in G}\sum_{q\in H}p q$.
\end{itemize}

When producing the polynomial $F$ from a circuit $\F$ we only use the
generic semiring axioms (i)--(iii) to write the result as a polynomial
(sum of monomials). For example, if $\F=x\pr (\vienas\su y)$ then
$F=x\su xy$, even though $\pol{F}=x$ in $\B$ and $\Min$, and
$\pol{F}=xy$ in $\Max$. It is thus important to note that the produced
by a given circuit $\F$ polynomial $F$ is the same over \emph{any}
semiring!

\begin{dfn}
  A circuit $\F$ \emph{computes} a polynomial $f$ if $\pol{F}=\pol{f}$
  ($F$ and $f$ coincide as functions).  A circuit $\F$ \emph{produces}
  $f$ if $F=f$ ($F$ and $f$ have the same set of monomials).
\end{dfn}
A circuit $\F$ \emph{simultaneously} computes (or produces) a given
set $\f$ of polynomials if, for every polynomial $f\in\f$, there is a
gate in $\F$ at which $f$ is computed (or produced).

When analyzing circuits, the following concept of ``parse graphs'' is
often useful.  A \emph{parse-graph} $\G$ in $\F$ is defined
inductively as follows: $\G$ includes the root (output gate) of
$\F$. If $\gate$ is a sum-gate, then exactly one of its inputs is
included in $\G$. If $\gate$ is a product gate, then both its input
gates are included in $\G$.  Note that each parse-graph produces
exactly one monomial in a natural way, and that each monomial $p\in F$
is produced by at least one parse-graph. If $p$ is multilinear, then
each parse-graph for $p$ is a tree.

\begin{itemize}

\item A circuit is \emph{homogeneous}, if polynomials produced at its
  gates are homogeneous. It is easy to see that a circuit is
  homogeneous if and only if the polynomial produced by it is
  homogeneous.

\item A circuit is \emph{multilinear}, if for every its product gate
  $\gate=v\pr w$, the sets of variables of the polynomials produced at
  gates $v$ and $w$ are disjoint. Sometimes, multilinear (in our
  sense) circuits are called also \emph{syntactically multilinear}.
\end{itemize}
Note that multilinear circuits can only compute multilinear
polynomials but, in general, circuits computing multilinear
polynomials need not be multilinear: this happens, for example, in
semirings $\B$ and~$\Min$. Still, Lemmas~\ref{lem:arithm2} and
\ref{lem:multilin} imply that this cannot happen in the remaining
three tropical semirings:

\begin{lem}\label{lem:multilin1}
  Every circuit computing a multilinear polynomial $f$ over $\R$,
  $\Max$, $\M$ or $\MAX$ must be multilinear. Moreover, over $\R$,
  $\M$ and $\MAX$, the circuit must even produce~$f$.
\end{lem}

We will be interested in the following two complexity measures of
polynomials~$f$, where the third measure is only for multilinear
polynomials:

\begin{itemize}
\item $\A(f)$ = minimum size of a circuit over semiring $\A$
  \emph{computing} $f$.

\item $\A[f]$ = minimum size of a circuit over semiring $\A$
  \emph{producing} $f$.

\item $\Mult{\A}{f}$ = minimum size of a \emph{multilinear} circuit
  over semiring $\A$ computing $f$.
\end{itemize}
What we are really interested in is the first measure $\A(f)$.  The
second measure $\A[f]$ is less interesting: it is the \emph{same} for
all semirings $\A$, because the formal polynomial of a given (fixed)
circuit is the same over all semirings. In particular, we have that
\[
\A[f]=\R[f]
\]
holds for every semiring $\A$ and every polynomial $f$.  Still, it
will be sometimes convenient \emph{not} to focus on the arithmetic
semiring $\R$ because the inequality $\A(f)\geq \A[f]$ is more
informative: it means that computing a given polynomial over $\A$ is
not easier than to produce this polynomial. This, for example, happens
in the arithmetic semiring $\R$: Lemma~\ref{lem:arithm0} implies
that $\R(f)\geq \R[f]$.

Also, Lemma~\ref{lem:multilin1} implies that the third measure
$\Mult{\A}{f}$ may be only interesting in semirings $\B$ and $\Min$:
if $\A\in\{\Max, \M, \MAX, \R\}$, then for every multilinear
polynomial $f$, we have that $\Mult{\A}{f}=\A(f)$.

\section{Some Polynomials}
\label{sec:important}

For the ease of reference, here we recall some polynomials which we
will use later to illustrate the lower bound arguments. Variables
$x_e$ of considered polynomials correspond to edges of $K_n$ or
$K_{n,n}$. Thus, monomials $\prod_{e\in E}x_e$ correspond to some
subgraphs $E$ of $K_{n}$ or $K_{n,n}$. Here are some of the
polynomials we will use later:

\begin{itemize}
\item Permanent polynomial $\PERM{n}$ = all perfect matchings in
  $K_{n,n}$.

\item Hamiltonian cycle polynomial $\HC{n}$ = all Hamiltonian cycles
  in $K_n$.

\item $k$-clique polynomial $\Clique{n,k}$ = all $k$-cliques in $K_n$.

\item Spanning tree polynomial $\ST{n}$ = all spanning trees in $K_n$
  rooted in node $1$.

\item $st$-connectivity polynomial $\STCONN{n}$ = all paths from $s=1$
  to $t=n$ in $K_n$.

\item All-pairs connectivity ``polynomial'' $\APSP{n}$ = \emph{set} of
  $\tbinom{n}{2}$ polynomials $\STCONN{n}$ corresponding to different
  pairs of start and target nodes $s$ and $t$.

\item Matrix product polynomial $\MP{n}$ = special case of $\APSP{n}$
  when only paths of length-$2$ are considered.
\item The connectivity polynomial $\CONN{n}$ = product of all
  polynomials of $\APSP{n}$.
\end{itemize}

In Section \ref{sec:rect} we will show that the first four polynomials
require $\Min$-circuits of exponential size, whereas the next result
shows that the last four polynomials all have $\Min$-circuits of
polynomial size. The following result---proved independently by
\Cite{Moore}{moore}, \Cite{Floyd}{floyd}, and
\Cite{Warshall}{warshall}---holds for every semiring with the
absorption axiom $a+ab=a$, including the boolean and $\Min$ semirings.

\begin{thm}[\cite{moore,floyd,warshall}]\label{thm:upper-bounds}
  Over semirings $\Min$ and $\B$, the polynomials of $\APSP{n}$ can
  all be simultaneously computed by a circuit of size $O(n^3)$.
\end{thm}

\begin{proof}
  Inputs for $\APSP{n}$ over the $\Min$ semiring are non-negative
  weights $x_{ij}$ of the edges of $K_n$. For every pair $i<j$ of
  distinct nodes of $K_n$, the goal is to compute the weight of the
  lightest path between $i$ and $j$; the weight of a path is the sum
  of weights of its edges.  The idea is to recursively compute the
  polynomials $\ff{f}{i,j}{k}$ for $k=0,1,\ldots,n$, whose value is
  the weight of the lightest walk between $i$ and $j$ whose all inner
  nodes lie in $[k]=\{1,\ldots,k\}$. Then $\ff{f}{i,j}{0}=x_{ij}$, and
  the recursion is: $\ff{f}{i,j}{k} = \min\Big\{\ff{f}{i,j}{k-1}, \
  \ff{f}{i,k}{k-1} + \ff{f}{k,j}{k-1}\Big\}$.  The output gates are
  $\ff{f}{i,j}{n}$ for all $i <j$. The total number of gates is
  $O(n^3)$.  Even though the circuit actually searches for weights of
  lightest \emph{walks}, it correctly computes $\APSP{}$ because every
  walk between two nodes $i$ and $j$ also contains a simple path (with
  no repeated nodes) between these nodes. Since the weights are
  non-negative, the minimum must be achieved on a simple path.  If we
  replace min-gates by OR-gates, and sum-gates by AND-gates, then the
  resulting circuit will compute $\APSP{n}$ over the boolean
  semiring~$\B$.
\end{proof}

\begin{remark}\label{rem:bellman}
  Earlier dynamic programming algorithm of \Cite{Bellman}{bellman} and
  \Cite{Ford}{ford} gives a (structurally) simpler $\Min$-circuit for
  $\STCONN{n}$. It tries to compute the polynomials $\ff{f}{j}{k}$
  whose value is the weight of the lightest walk between $1$ and $j$
  with at most $k$ edges. Then $\ff{f}{j}{1}=x_{1j}$, and the
  recursion is: $\ff{f}{j}{k}$ = the minimum of $\ff{f}{j}{k-1}$ and
  of $\ff{f}{i}{k-1} + x_{i,j}$ over all nodes $i\neq j$. The output
  gate is $\ff{f}{n}{n-1}$. The circuit also has $O(n^3)$ fanin-$2$
  gates.
\end{remark}

\begin{remark}\label{rem:ST}
  Theorem~\ref{thm:upper-bounds} immediately implies that the
  polynomials $\MP{n}$, $\CONN{n}$, and $\STCONN{n}$ can also be
  computed by $\Min$-circuits of size $O(n^3)$. Moreover, over the
  boolean semiring, the spanning tree polynomial $\ST{}$ represents
  the same boolean function as $\CONN{}$.  Thus,
  Theorem~\ref{thm:upper-bounds} also gives $\B(\ST{n})=O(n^3)$.
\end{remark}

\begin{table}
  \begin{center}
    \begin{tabular*}{0.9\textwidth}{l@{\hskip 1cm}l@{\hskip 1cm}l}
      {\bf Polynomial} $f$ &  {\bf Bound} & {\bf Reference}\\
      \noalign{\smallskip}
      $\ST{n}$ & $\B(f)=O(n^3)$, $\A(f)=2^{\Omega(n)}$ & Rem.~\ref{rem:ST}, Thm.~\ref{thm:bounds1} \\[0.5ex]
      $\CONN{n}$, $\STCONN{n}$ & $\Min(f)=O(n^3)$, $\R[f]\geq \Max(f)=2^{\Omega(n)}$ & Rem.~\ref{rem:ST} \\[0.5ex]
      $\APSP{n}$, $\MP{n}$ & $\Min(f)=\Theta(n^3)$ & Cor.~\ref{cor:APSP}  \\[0.5ex]
      $\PERM{n}$, $\HC{n}$ & $\A(f)=2^{\Omega(n)}$ & Thm.~\ref{thm:bounds1} \\[0.5ex]
      $\Clique{n,k}$ & $\A(f)\geq \binom{n}{k}-1$ &  Cor.~\ref{cor:cliq}\\
      \noalign{\smallskip}
    \end{tabular*}
    \caption[]{Summary of specific bounds; $\A(f)$ stands for any of
      $\Min(f)$, $\Max(f)$ and $\Mult{\B}{f}$.}
    \label{tab:bounds}
  \end{center}
\end{table}

In the rest of the paper, we will present various lower bound argument
for tropical circuits. Table~\ref{tab:bounds} summarizes the resulting
specific bounds obtained by these arguments for the polynomials listed
above.

\section{Reduction to the Boolean Semiring}
\label{sec:boolean}

A semiring $\A=(S,\su,\pr,\nulis,\vienas)$ is of
\emph{zero-characteristic}, if $\vienas\su\vienas \su \cdots \su
\vienas\neq \nulis$ holds for any finite sum of the unity~$\vienas$.
Note that all semirings we consider are of zero-characteristic.  The
following seems to be a ``folklore'' observation.

\begin{lem}\label{lem:bool}
  If a semiring $\A$ is of zero-characteristic, then $\A(f)\geq \B(f)$
  holds for every polynomial~$f$.
\end{lem}

\begin{proof} Let $\F$ be a circuit over $\A$ computing a given
  polynomial $f$.  The circuit must correctly compute $f$ on any
  subset of the domain~$S$.  We choose the subset $S_+=\{\nulis,
  \up{1}, \up{2},\ldots\}$, where $\up{n}=\vienas\su \cdots \su
  \vienas$ is the $n$-fold sum of the multiplicative unit element
  $\vienas$.  Note that $\up{n}\neq\nulis$ holds for all $n\geq 1$,
  because $\A$ has zero-characteristic.

  Since $\up{n}\su \up{m}=\up{n+m}$ and $\up{n}\pr \up{m}=\up{n\cdot
    m}$, $\A_+=(S_+,\su,\pr,\nulis,\vienas)$ is a semiring. Since
  $S_+\subseteq S$, the circuit must correctly compute $f$ over this
  semiring as well. But the mapping $h:S_+\to\{0,1\}$ given by
  $h(\nulis)=0$ and $h(\up{n})=1$ for all $n\geq 1$, is a homomorphism
  from $\A_+$ into the boolean semiring $\B$ with $h(x\su y)=h(x)\lor
  h(y)$ and $h(x\pr y)=h(x)\land h(y)$.  So, if we replace each
  $\su$-gate by a logical OR, and each $\pr$-gate by a logical AND,
  then the resulting monotone boolean circuit computes the polynomial
  $f$ over~$\B$.
\end{proof}

\begin{remark}
  One can easily show that, if the input variables can only take
  boolean values $0$ and $1$, then $\Min(f)\leq 2\cdot \B(f)$ holds
  for every multilinear polynomial. Indeed, having a (boolean) circuit
  $\F$ for $f$, just replace each AND gate $u\land v$ by a Min gate
  $\min(u,v)$, and each OR gate $u\lor v$ by $\min(1,u+v)$.  The point
  however is that tropical circuits must work correctly on much larger
  domain than $\{0,1\}$.  This is why lower bounds for tropical
  circuits do not translate to lower bounds for monotone boolean
  circuits.  And indeed, there are explicit polynomials $f$, as the
  spanning tree polynomial $f=\ST{n}$, such that $\B(f)=O(n^3)$ but
  $\Min(f)=2^{\Omega(n)}$; the upper bound is shown in
  Remark~\ref{rem:ST}, and the lower bound will be shown in
  Theorem~\ref{thm:bounds1}.
\end{remark}

To prove lower bounds in the boolean semiring---and hence, by
Lemma~\ref{lem:bool}, also in tropical semirings---one can try to use
the following general lower bounds criterion proved in \cite{Juk99}
(see also \cite[Sect. 9.4]{myBFC-book} for a simplified proof).

For $a\in \{0,1\}$, an $a$-\emph{term} of a monotone boolean function
is a subset of its variables such that, when all these variables are
fixed to the constant $a$, the function outputs value $a$, independent
of the values of other variables. It is easy to see that every
$0$-term must intersect every $1$-term, and vice versa.  Say that a
family of sets $A$ \emph{covers} a family of sets $B$ if every set in
$B$ contains at least one set of~$A$.

 \begin{dfn}\label{def:simple}
   A monotone boolean function $f(X)$ of $|X|=n$ variables is
   $t$-\emph{simple\/} if for all integers integers $2\leq r,s\leq n$,
   such that
   \begin{description}
   \item[(i)] either the set of all $0$-terms of $f$ can be covered by
     $t(r-1)^s$ $s$-element subsets of~$X$,
   \item[(ii)] or the set of all $1$-terms of $f$ can be covered by at
     most $t(s-1)^r$ $r$-element subsets of $X$ plus $s-1$ single
     variables.
   \end{description}
 \end{dfn}
 Note that this ``asymmetry'' between (i) and (ii) (allowing
 additional $s-1$ single variables in a cover) is important: say,
 condition (i) is trivially violated, if $f$ contains a $0$-term
 $T=\{x_1,\ldots,x_k\}$ with $k<s$.  But then (ii) is satisfied,
 because $T$ must intersect all $1$-terms, implying that the single
 variables $x_1,\ldots,x_k$ cover all of them.

 \begin{thm}[\cite{Juk99}]
   \label{thm:crit}
   If $f$ is not $t$-simple, then $\B(f)>t$.
 \end{thm}

 \section{Reduction to the Arithmetic Semiring}
 \label{sec:to-artihm}

 As we already mentioned in the introduction, circuits over the
 arithmetic semiring $\R$ are no more powerful than circuits over
 boolean or tropical semirings.  The weakness of circuits computing a
 given polynomial $f$ over $\R$ lies in the fact (following from
 Lemma~\ref{lem:arithm0}) that they cannot produce any ``redundant''
 monomials, those not in $f$. That is, here we have $\R(f)\geq \R[f]$.
 On the other hand, if the semiring $\A$ is additively-idempotent,
 then
 \begin{equation}\label{eq:not-larger}
   \A(f)\leq \A[f]=\R[f]\,.
 \end{equation}
 This holds because in an additively-idempotent semiring $\A$ (where
 $x\su x=x$ holds), the multiplicities of monomials have no effect on
 the represented function. But, in general, we have no converse
 inequality $ \A(f)\geq \R[f]$: for some polynomials $f$, $\R[f]$ may
 be even exponentially larger than $\A(f)$. Such is, for example, the
 $st$-connectivity polynomial $f=\STCONN{n}$.  For this polynomial, we
 have $\Min(f)=O(n^3)$ (see Remark~\ref{rem:ST}), but it is relatively
 easy to show that $\Min[f]=\R[f]=2^{\Omega(n)}$ (see
 Theorem~\ref{thm:bounds2} below). We will now show that the reason
 for such a large gap is the non-homogeneity of $\STCONN{}$.

 Following \Cite{Jerrum and Snir}{jerrum}, define the \emph{lower
   envelope} of a polynomial $f$ to be the polynomial $\lenv{f}$
 consisting of all monomials of $f$ of smallest degree.  Similarly,
 the \emph{higher envelope}, $\henv{f}$, of $f$ consists of all
 monomials of $f$ of largest degree. Note that both polynomials
 $\lenv{f}$ and $\henv{f}$ are homogeneous, and $\lenv{f}=\henv{f}=f$,
 if $f$ itself is homogeneous.

 \begin{obs}\label{fact:1a}
   If a polynomial $f$ can be produced by a circuit of size $s$, then
   both $\lenv{f}$ and $\henv{f}$ can be produced by homogeneous
   circuits of size~$s$.
 \end{obs}

\begin{proof}
  Take a circuit producing $f$.  The desired homogeneous sub-circuits
  producing the lower or the higher envelope can be obtain by starting
  with input gates, and removing (if necessary) one of the wires of
  every sum-gate, at inputs of which polynomials of different degrees
  are produced.
\end{proof}

 \begin{thm}\label{thm:homog}
   For every multilinear polynomial $f$, we have
   \begin{equation}\label{eq:lower-bounds}
     \R[f]\geq \Mult{\B}{f}\geq \Min(f) \geq \R[\lenv{f}]\ \ \mbox{ and }\ \
     \R[f]\geq \Max(f)\geq \R[\henv{f}]\,.
   \end{equation}
   If $f$ is also homogeneous, then
   \[
   \Mult{\B}{f}=\Min(f)=\Max(f)=\R[f]\,.
   \]
 \end{thm}

\begin{proof}
  By \eqref{eq:not-larger}, we only have to prove the lower bounds
  \eqref{eq:lower-bounds}.  To prove that $\Mult{\B}{f}\geq \Min(f)$,
  let $\F$ be a multilinear monotone boolean circuit computing
  $f$. Since the circuit is multilinear, its produced polynomial $F$
  is also multilinear. Since every monotone boolean function has a
  \emph{unique} shortest monotone DNF, this implies that
  $\lmin{F}=\lmin{f}$. Since $f$ and $\lmin{f}$ represent the same
  function over $\Min$, the circuit $\F$ with OR gates replaced by Min
  gates, and AND gates by Sum gates will compute $f$ over $\Min$.

  To prove the inequality $\Min(f)\geq \R[\lenv{f}]$, take a minimal
  circuit $\F$ over $\Min$ computing $f$.  Observation~\ref{fact:1a}
  implies that the lower envelope $\lenv{F}$ of the polynomial $F$
  produced by $\F$ can be also produced by a (homogeneous) circuit of
  size at most $\size{\F}$.  Hence, $\R[\lenv{F}]\leq
  \size{\F}=\Min(f)$. On the other hand, Lemma~\ref{lem:arithm1}
  implies that $\lenv{f}=\lenv{F}$, and we are done.

  The proof of $\Max(f)\geq \R[\henv{f}]$ is the same by using
  Lemma~\ref{lem:arithm2}. 
\end{proof}

The second claim of Theorem~\ref{thm:homog} has an important
implication concerning the power of dynamic programs, which can be
roughly stated as follows:
\begin{enumerate}
\item[] For optimization problems whose target polynomials are
  \emph{homogeneous}, dynamic programming is no more powerful than
  monotone arithmetic circuits!
\end{enumerate}

\section{Relative Power of Semirings}
\label{sec:gaps}

The reductions to the boolean and to the arithmetic semirings
(Lemma~\ref{lem:bool} and Theorem~\ref{thm:homog}) give us the
following relations for every multilinear polynomial $f$:
\[
\B(f) \leq \Min(f) \leq \Mult{\B}{f} \leq \M(f)=\R[f]
\]
and
\[
\B(f)\leq \Max(f) \leq \MAX(f)= \R[f]\,.
\]
If, additionally, $f$ is also homogeneous, then
\[
\B(f)\leq \Mult{\B}{f} = \Min(f)=\Max(f)=\M(f)= \MAX(f)=\R[f]\,.
\]
Moreover, all inequalities are strict: for some polynomials $f$, one
side can be even exponentially smaller than the other. Moreover, the
$\Max/\Min$ and $\Min/\Max$ gaps can be also exponential.

To show that circuits over the tropical semirings can be exponentially
weaker than those over the boolean semiring, consider the the spanning
tree polynomial $f=\ST{n}$ and the graph connectivity polynomial
$g=\CONN{n}$.  Over the boolean semiring $\B$, these polynomials
represent the same boolean function: a graph is connected if and only
if it has a spanning tree.  This gives $\B(f)=\B(g)$ and
$\Mult{\B}{f}=\Mult{\B}{g}$.  Moreover, we already know (see
Remark~\ref{rem:ST}) that $\B(g)=O(n^3)$ and $\Min(g)=O(n^3)$. On the
other hand, a relatively simple argument (the ``rectangle bound'')
yields $\R[f]=2^{\Omega(n)}$ (see Theorem~\ref{thm:bounds1}
below). Since the polynomial $f$ is homogeneous,
Theorem~\ref{thm:homog} implies that $\Min(f)$, $\Max(f)$ and
$\Mult{\B}{f}$ coincide with $\R[f]$, and hence, are also exponential
in $n$. We thus have gaps:
\begin{align*}
  \Min(f)/\B(f),\  \Max(f)/\B(f)&= 2^{\Omega(n)}\ \ \mbox{ for $f=\ST{n}$;}\\
  \Mult{\B}{g}/\Min(g),\ \Mult{\B}{g}/\B(g)&= 2^{\Omega(n)}\ \ \mbox{
    for $g=\CONN{n}$.}
\end{align*}
The latter gap $\Mult{\B}{g}/\B(g)= 2^{\Omega(n)}$ also shows that
there is no ``multilinear version'' of the Floyd--Warshall algorithm,
even in the boolean semiring.

To show that the remaining gaps can also be exponential, it is enough
to take \emph{any} multilinear and homogeneous polynomial
$f(x_1,\ldots,x_n)$ such that $\R[f]$ is exponential in $n$, and to
consider its two ``saturated'' versions $\lsat{f}$ and $\hsat{f}$,
where $\lsat{f}$ is obtained by adding to $f$ all $n$ monomials
$x_1,x_2,\ldots,x_n$ of degree $1$, and $\hsat{f}$ is obtained by
adding to $f$ the monomial $x_1x_2\cdots x_n$ of degree~$n$.

\begin{lem}
  Let $f(x_1,\ldots,x_n)$ be a multilinear and homogeneous polynomial.
  Then both $\Min(\hsat{f})$ and $\Max(\lsat{f})$ are at least
  $\R[f]$, but all $\Max(\hsat{f})$, $\Min(\lsat{f})$ and
  $\Mult{\B}{\lsat{f}}$ are at most $n$.
\end{lem}

\begin{proof}
  Since $f$ is the lower envelope of $\hsat{f}$, and the higher
  envelope of $\lsat{f}$. Theorem~\ref{thm:homog} implies that
  $\Min(\hsat{f})\geq \R[f]$ and $\Max(\lsat{f})\geq \R[f]$.  On the
  other hand, over the $\Max$-semiring, the polynomial $\hsat{f}$
  computes $x_1+x_2+\cdots+x_n$, whereas over the $\Min$-semiring,
  $\lsat{f}$ computes $\min\{x_1,x_2,\ldots,x_n\}$, and computes
  $x_1\lor x_2\lor \cdots \lor x_n$ over the boolean semiring.  Hence,
  all $\Max(\hsat{f})$, $\Min(\lsat{f})$ and $\Mult{\B}{\lsat{f}}$ are
  at most $n$.
\end{proof}

Since, there are many linear and homogeneous polynomials requiring
mo\-no\-to\-ne arithmetic circuits of exponential size (see,
e.g. Table~\ref{tab:bounds}), the saturated versions of $f$
immediately give exponential gaps.

Still, the ``saturation trick'' leads to somewhat artificial examples,
and it would be interesting to establish exponential gaps using
``natural'' polynomials.  For example, the $\Max/\Min$ gap is achieved
already on a very natural $st$-con\-nec\-ti\-vi\-ty polynomial
$h=\STCONN{n}$. We know that $\Min(h)=O(n^3)$ (Remark~\ref{rem:ST}),
but a simple argument (see Theorem~\ref{thm:bounds2}) shows that
$\Max(h)=2^{\Omega(n)}$. Hence,
\[
\Max(h)/\Min(h) = 2^{\Omega(n)}\ \ \mbox{ for $h=\STCONN{n}$.}
\]

From now on we concentrate on the lower bound \emph{arguments}
themselves.

\section{Lower Bounds for Separated Polynomials}
\label{sec:schnorr}

Let $g(x_1,\ldots,x_n)$ be a polynomial in $n\geq 3$ variables.  An
\emph{enrichment} of $g$ is a polynomial $h$ in $n-1$ variables
obtained by taking some variable $x_k$ and replacing it by a sum $x_i+
x_j$ or by a product $x_i x_j$ of some other two (not necessarily
distinct) variables, where $k\not\in\{i,j\}$.  A \emph{progress
  measure} of polynomials is an assignment of non-negative numbers
$\bound{g}$ to polynomials $g$ such that
\begin{description}
\item[(i)] $\bound{x_i}=0$ for each variable $x_i$;
\item[(ii)] $\bound{h}\leq \bound{g}+1$ for every enrichment $h$ of
  $g$.
\end{description}

\begin{lem}\label{lem:schnorr1}
  For every polynomial $f$, and every progress measure $\bound{f}$, we
  have $\R[f]\geq \bound{f}$.
\end{lem}

\begin{proof}
  Take a monotone arithmetic circuit $\F$ with $s=\R[f]$ gates
  producing $f$. We argue by induction on $s$.  If $s=0$, then
  $\F=x_i$ in an input variable, and we have $\R[f]=0=\bound{f}$.  For
  the induction step, take one gate $\gate=x_i\ast x_j$ where
  $\ast\in\{+,\cdot\}$. Let $\F'(x_1,\ldots,x_n,y)$ be the circuit
  with the gate $\gate$ replaced by a new variable $y$.  Hence,
  $\size{\F'}=\size{\F}-1$ and $F(x_1,\ldots,x_n)$ is an enrichment of
  $F'(x_1,\ldots,x_n,y)$. By the induction hypothesis, we have that
  $\size{\F'}\geq \bound{F'}$.  Together with $\bound{F}\leq
  \bound{F'}+1$, this yields $\size{\F}=\size{\F'}+1\geq
  \bound{F'}+1\geq \bound{F}$.
\end{proof}
Recall that a monomial $p$ \emph{contains} a monomial $q$ (as a
factor), if $p=qq'$ for some monomial $q'$.

\begin{dfn}\label{def:separated}
  A sub-polynomial $P\subseteq f$ is \emph{separated} if the product
  $pq$ of any two monomials $p$ and $q$ of $P$ contains no monomial of
  $f$ distinct from $p$ and from $q$.  Let
  \[
  \bbound{f}:=\max\{|P|-1\colon \mbox{$P\subseteq f$ is
    separated}\}\,.
  \]
\end{dfn}
Note that we consider separateness \emph{within} the entire set $f$ of
monomials: it is not enough that the product $pq$ contains no third
monomial of~$P$---it must not contain any third monomial of the entire
polynomial~$f$.

Note also that a multilinear polynomial $f$ of minimum degree $m$ is
separated, if every monomial of $f$ is uniquely determined by any
subset of $\lceil m/2\rceil$ its variables. (Being uniquely determined
means that no other monomial contains the same subset of variables.)
Indeed, if $p\pr q$ contains some monomial $r$ then $r$ and $p$ (or
$r$ and $q$) must share at least $\lceil m/2\rceil$ variables,
implying that $r=p$ (or $r=q$) must hold.

\begin{thm}[\Cite{Schnorr}{schnorr}]\label{thm:schnorr}
  For every polynomial $f$, we have $\R[f]\geq \bbound{f}$, where
  \[
  \bbound{f}:=\max\{|P|-1\colon \mbox{$P\subseteq f$ is
    separated}\}\,.
  \]
  In particular, $\R[f]\geq |f|-1$ if the polynomial $f$ itself is
  separated.
\end{thm}

\begin{proof}
  It is enough to show that the measure $\bbound{f}$ is a progress
  measure. The first condition (i) is clearly fulfilled, since
  $\bbound{x_i}=1-1=0$.  To verify the second condition (ii), let
  $f(x_1,\ldots,x_n,y)$ be a polynomial, and $h(x_1,\ldots,x_n)$ be
  its enrichment.  Our goal is to show that $\bbound{f}\geq
  \bbound{h}-1$.  We only consider the ``hard'' case when $y$ is
  replaced by a sum of variables:
  $h(x_1,\ldots,x_n)=f(x_1,\ldots,x_n,u+ v)$, where
  $u,v\in\{x_1,\ldots,x_n\}$.

  To present the proof idea, we first consider the case when no
  monomial of $f$ contains more than one occurrence of the variable
  $y$. Then every monomial $yp$ of $f$ turns into two monomials $up$
  and $vp$ of $h$.  To visualize the situation, we may consider the
  bipartite graph $G\subseteq f\times h$, where every monomial $yp\in
  f$ is connected to two monomials $up,vp\in h$; each monomial $q\in
  f$ without $y$ is connected to $q\in h$.  Take now a separated
  subset $P\subseteq h$ such that $|P|-1=\bbound{h}$, and let
  $Q\subseteq f$ be the set of its neighbors in $G$. Our goal is to
  show that:
  \begin{description}
  \item[(a)] $|Q|\geq |P|-1$, and
  \item[(b)] $Q$ is separated.
  \end{description}
  Then the desired inequality $\bbound{f}\geq |Q|-1\geq
  |P|-2=\bbound{h}-1$ follows.

  To show item (a), it is enough to show that at most one monomial in
  $Q$ can have both its neighbors in $P$. To show this, assume that
  this holds for some two monomials $yp$ and $yq$ of $Q$. Then all
  four monomials $up,vp,uq,vq$ belong to $P$. But this contradicts the
  separateness of $P$, because the product $up\pr vq$ contains the
  third monomial $uq$ (and $vp$).

  To show item (b), assume that the product $p\pr q$ of some two
  monomials $p\neq q$ of $Q$ contains some third monomial $r\in h$.
  Let $p',q'\in P$ be some neighbors of $p$ and $q$ lying in $P$.
  Then the product $p'\pr q'$ must contain one (of the two) neighbors
  of $r$.  Since \emph{both} of these neighbors of $r$ belong to $h$,
  we obtain a contradiction with the separateness of~$P$.

  In general (if $y$ can have any degrees in $f$), a monomial $y^kp$
  of $f$ has $k+1$ neighbors $u^iv^{k-i}p$, $i=0,1,\ldots,k$ in $h$.
  To show (a), it is again enough to show that at most one monomial in
  $Q$ can have two neighbors in $P$.  For this, assume that there are
  two monomials $p\neq q$ such that all four monomials
  $u^av^{k-a}p,u^bv^{k-b}p, u^cv^{l-c}q,u^dv^{l-d}q$ belong to
  $P$. Assume w.l.o.g. that $a=\max\{a,b,c,d\}$.  Then the product
  $u^av^{k-a}p\pr u^cv^{l-c}q$ contains $u^av^{l-c}q$, and (since
  $c\leq a$) contains the monomial $u^av^{l-a}q$ of $h$, contradicting
  the separateness of~$P$.  The proof of (b) is similar.
\end{proof}

\begin{remark}
  It is not difficult to see that we have a stronger inequality
  $\bbound{f}\geq \bbound{h}$, if the variable $y$ is replaced by the
  product $uv$ (instead of the sum $u+v$).  Thus, in fact,
  Theorem~\ref{thm:schnorr} gives a lower bound on the number of sum
  gates.
\end{remark}

As a simple application of Schnorr's argument, consider the
\emph{triangle polynomial}
\[
\tR{n}(x,y,z)=\sum_{i,j,k\in[n]} x_{ik}y_{kj}z_{ij}\,.
\]
This polynomial has $3n$ variables and $n^3$ monomials.

\begin{cor}\label{cor:triangle}
  If $f=\tR{n}$, then $\Min(f)=\Max(f)=\R[f]=\Theta(n^3)$.
\end{cor}

\begin{proof}
  The equalities $\Min(f)=\Max(f)=\R[f]$ hold by
  Theorem~\ref{thm:homog}, because $f$ is multilinear and
  homogeneous. The upper bound $\R[f]=O(n^3)$ is trivial.  To prove
  the lower bound $\R[f]=\Omega(n^3)$, observe that every monomial
  $p=x_{ik}y_{kj}z_{ij}$ of $f$ is uniquely determined by any choice
  of any two of its three variables. This implies that $p$ cannot be
  contained in a union of any two monomials distinct from $p$. Thus,
  the polynomial $f$ is separated, and its Schnorr's measure is
  $\bbound{f}=n^3-1$.  Theorem~\ref{thm:schnorr} yields $\R[f]\geq
  \bbound{f}=n^3-1$, as desired.
\end{proof}

Recall that the $k$-clique polynomial $\Clique{n,k}$ has
$\binom{n}{k}$ monomials $\prod_{i< j\in S}x_{ij}$ corresponding to
subsets $ S\subseteq [n]$ of size $|S|=k$.  This is a homogeneous
multilinear polynomial of degree~$\binom{k}{2}$. Note that $\tR{n}$ is
a sub-polynomial of $\Clique{3n,3}$ obtained by setting some variables
to~$0$.

By Lemma~\ref{lem:bool}, an exponential lower bound for $\Clique{n,s}$
over the tropical $\Min$ follows from Razborov's lower bound for this
polynomial over the boolean semiring $\B$ \cite{razb1}. However, the
proof over $\B$ is rather involved. On the other hand, in tropical
semirings such a bound comes quite easily.

\begin{cor}\label{cor:cliq}
  For $f=\Clique{n,k}$, $\Min(f)$, $\Max(f)$ and $\Mult{\B}{f}$ are at
  least $\binom{n}{k}-1$.
\end{cor}
This lower bound on $\Mult{\B}{f}$ was proved by
\Cite{Krieger}{krieger} using different arguments.

\begin{proof}
  Since $f$ is multilinear and homogeneous, it is enough (by
  Theorems~\ref{thm:homog}) to show the corresponding lower bound on
  $\R[f]$. By Theorem~\ref{thm:schnorr}, it is enough to show that $f$
  is separated.

  Assume for the sake of contradiction, that the union of two distinct
  $k$-cliques $A$ and $B$ contains all edges of some third clique
  $C$. Since all three cliques are distinct and have the same number
  of nodes, $C$ must contain a node $u$ which does not belong to $A$
  and a node $v$ which does not belong to $B$. This already leads to a
  contradiction because either the node $u$ (if $u = v$) or the edge
  $\{u, v\}$ (if $u\neq v$) of $C$ would remain uncovered by the
  cliques $A$ and~$B$.
\end{proof}

Recall that the dynamic programming algorithm of Floyd--Warshall
implies that the all-pairs shortest path polynomial $\APSP{n}$, and
hence, also the matrix product polynomial $\MP{n}$, have
$\Min$-circuits of size $O(n^3)$; see Theorem~\ref{thm:upper-bounds}.
On the other hand, using Theorem~\ref{thm:schnorr} one can show that
this algorithm is optimal: a cubic number of gates is also necessary.

\begin{cor}\label{cor:APSP}
  Both $\Min(\APSP{n})$ and $\Min(\MP{n})$ are $\Theta(n^3)$.
\end{cor}

\begin{proof}
  It is enough to show that $\Min(\MP{n})=\Omega(n^3)$. Recall that
  $\MP{n}(x,y)$ is the set of all $n^2$ polynomials
  $f_{ij}=\sum_{k\in[n]} x_{ik}y_{kj}$.  Since the triangle polynomial
  $\tR{n} = \sum_{i,j\in [n]} z_{ij}f_{ij}$ is just a single-output
  version of $\MP{n}$, and its complexity is by at most an additive
  factor of $2n^2$ larger than that of $\MP{n}$, the desired lower
  bound for $\MP{n}$ follows directly from
  Corollary~\ref{cor:triangle}.
\end{proof}

\Cite{Kerr}{kerr} earlier proved $\Min(\MP{n})=\Omega(n^3)$ using a
different argument, which essentially employs the fact the
$\Min$-semiring contains more than two distinct elements.  Since this
``domain-dependent'' argument may be of independent interest, we
sketch it.

\begin{proof} (Due to \Cite{Kerr}{kerr}) Let $\F$ be a $\Min$-circuit
  computing all $n^2$ polynomials
  \[
  f_{ij}(x)=\min\{x_{ik}+y_{kj}\colon k=1,\ldots,n\}\,.
  \]
  By Lemma~\ref{lem:arithm1}, for each polynomial $f_{ij}$ there must
  be a gate $\gate_{ij}$, the polynomial $F_{ij}$ produced at which is
  of the form $F_{ij}=\min\{f_{ij}, G_{ij}\}$, where $G_{ij}$ is some
  set of monomials (sums), each containing at least one monomial of
  $f_{ij}$.

  Assign to every monomial $p=x_{ik}+y_{kj}$ of $f_{ij}$ a sum gate
  $\gate_p$ with the following two properties: (i) $p$ is produced at
  $\gate_p$, and (ii) there is a path from $\gate_p$ to $\gate_{ij}$
  containing no sum gates. Since $a+a=a$ does not hold in $\Min$, at
  least one such gate must exist for each of the monomials
  $x_{ik}+y_{kj}$.

  It remains therefore to show that no other term $x_{ab}+y_{bc}$ gets
  the same gate $\gate_p$.  To show this, assume the opposite. Then at
  the gate $\gate_p$ some sum
  \[
  \min\{x_{ik}, \alpha, \ldots\}+\min\{y_{kj}, \ldots\}
  \]
  is computed, where $\alpha\in\{x_{ab}, y_{bc}\}$ is a single
  variable distinct from $x_{ik}$ and $y_{kj}$.  Set $\alpha:=0$,
  $x_{ik}=y_{kj}:=1$ and set all remaining variables to $2$.  Then the
  first minimum in the sum above evaluates to $0$, and we obtain
  $\pol{F}_{ij}(x)\leq 1$.  But $\pol{f}_{ij}(x)=2$ because the term
  $x_{ik}+y_{kj}$ gets value $1+1=2$, and the remaining terms of
  $f_{ij}$ get values $\geq 2+0=2$. This gives the desired
  contradiction.
\end{proof}

\begin{remark}\label{rem:paterson}
  Using more subtle arguments, \Cite{Paterson}{paterson}, and
  \Cite{Mehlhorn and Galil}{mehlhorn} succeeded to prove a cubic lower
  bound $\Omega(n^3)$ for $\MP{n}$ even over the boolean
  semiring~$\B$.
\end{remark}

\begin{remark}
  The argument used by Schnorr \cite{schnorr} is inductive, and is
  currently known as the \emph{gate-elimination method}. Having a
  circuit $\F$ of $n$ variables, replace its first gate by a new
  variable, use induction hypothesis for the resulting circuit $\F'$
  of $n+1$ variables but of smaller size to make a desired conclusion
  about the original circuit $\F$.  Using a similar gate-elimination
  reasoning, Baur and Strassen \cite{BS83} proved the following
  surprising upper bound: if a polynomial $f(x_1,\ldots,x_n)$ can be
  produced by a circuit of size $s$, then the polynomial $f$ and all
  its $n$ partial derivatives $\Der{f}{x_i}$ $(i=1,\ldots,n)$ can all
  be simultaneously produced by a circuit of size only $4n$. (Note
  that a trivial upper bound is about $sn$.) Their (relatively simple)
  argument uses gate-elimination together with the chain rule for
  partial derivatives. If the polynomial $f$ is multilinear, then
  $\Der{f}{x_i}$ is a polynomial obtained from $f$ by removing all
  monomials not containing $x_i$, and removing $x_i$ from all
  remaining monomials.  In particular, if a sum $f=\sum_{i=1}^k y_i
  f_i(x_1,\ldots,x_n)$ can be produced by a circuit of size $s$, then
  all polynomials $f,f_1,\ldots,f_k$ can be simultaneously produced by
  a circuit of size~$4s$.
\end{remark}

\section{Decompositions and Cuts}
\label{sec:decomp}

Besides the gate-elimination method, most of lower bound arguments for
monotone arithmetic circuits follow the following general frame: if a
polynomial $f$ can be produced by a circuit of size $s$, then $f$ can
be written as a sum $f=\sum_{i=1}^t g_i$ of $t=O(s)$ ``rectangles''
$g_i$. Usually, these ``rectangles'' $g_i$ are products of two (or
more) polynomials of particular degrees. Let us first explain, where
these ``rectangles'' come from.

Let $\F$ be a circuit over some semiring
$\A=(S,\su,\pr,\nulis,\vienas)$.  For a gate $\gate$ in $\F$, let
$\Prod{\gate}$ denote the polynomial produced at $\gate$, and let
$\F_{\gate=\nulis}$ denote the circuit obtained from $\F$ by replacing
the gate $\gate$ by the additive identity $\nulis$.  Recall that
$a\pr\nulis=\nulis$ holds for all $a\in S$.  Hence, the polynomial
$F_{\gate=\nulis}$ produced by $\F_{\gate=\nulis}$ consists of only
those monomials of $F$ which do not ``use'' the gate $\gate$ for their
production. To avoid trivialities, we will always assume that
$F_{\gate=\nulis}\neq F$, i.e. that there are no ``redundant'' gates.

\begin{lem}\label{lem:contain}
  For every gate $\gate$ in $\F$, the polynomial $F$ produced by $\F$
  can be written as a sum $F=F_{\gate} \su F_{\gate=\nulis}$ of two
  polynomials, the first of which has the form
  $F_{\gate}=\Prod{\gate}\pr \ext{\gate}$ for some polynomial
  $\ext{\gate}$.
\end{lem}

\begin{proof}
  If we replace the gate $\gate$ by a new variable $y$, the resulting
  circuit produces a polynomial of the form $y\pr A+F_{\gate=\nulis}$
  for some polynomial $A$.  It remains to substitute all occurrences
  of the variable $y$ with the polynomial $\Prod{\gate}$ produced at
  the gate $\gate$.
\end{proof}

\begin{remark}\label{rem:double-counting}
  Roughly speaking, the number $|F_{\gate}|$ of monomials in the
  polynomial $F_{\gate}$ is the ``contribution'' of the gate $\gate$
  to the production of the entire polynomial $F$. Intuitively, if this
  contribution is small for many gates, then there must be many gates
  in $\F$. More formally, associate with each monomial $p\in F$ some
  of its parse-graphs $\F_p$ in $\F$.  Observe that $u\in \F_p$
  implies $p\in F_u$.  Thus, double-counting yields
  \[
  \size{\F}=\sum_{u\in \F}1 \geq \sum_{u\in\F}\sum_{p\in F\colon u\in
    \F_p}\frac{1}{|F_u|} =\sum_{p\in F}\sum_{u\in \F_p}\frac{1}{|F_u|}
  \geq |F|\cdot \min_{p\in F} \sum_{u\in \F_p}\frac{1}{|F_u|}\,.
  \]
  So, in principle, one can obtain strong lower bounds on the total
  number of gates in $\F$ by showing that this latter minimum cannot
  be too small.
\end{remark}

The polynomial $\ext{\gate}$ in Lemma~\ref{lem:contain} can be
explicitly described by associating polynomials with paths in the
circuit $\F$.  Let $\pi$ be a path from a gate $\gate$ to the output
gate, $u_1,\ldots,u_m$ be all product gates along this path (excluding
the first gate $\gate$, if it itself is a product gate), and
$w_1,\ldots,w_m$ be input gates to these product gates \emph{not
  lying} on the path $\pi$. We associate with $\pi$ the polynomial
$\Prod{\pi}:=\Prod{w_1}\pr \Prod{w_2}\pr \cdots\pr \Prod{w_m}$.  Then
\[
\ext{\gate}=\sum_{\pi} \Prod{\pi}\,,
\]
where the sum is over all paths $\pi$ from $\gate$ to the output gate.

Lemma~\ref{lem:contain} associates sub-polynomials $\Prod{\gate}\pr
\ext{\gate}$ of $F$ with \emph{nodes} (gates) $\gate$ of $\F$.  In
some situations, it is more convenient to associate sub-polynomials
with \emph{edges}.  For this, associate with every edge $(u,v)$, where
$v=u\ast w$ is some gate with $\ast\in\{\su,\pr\}$ of $\F$, the
polynomial
\[
\Ext{v}{u}:=A\pr \ext{v}\ \ \mbox{ where }\ \ A=\begin{cases}
  \vienas & \mbox{ if $\ast=\su$;}\\
  \Prod{w} & \mbox{ if $\ast=\pr$.}
\end{cases}
\]
That is, $\Ext{v}{u}=\ext{v}$ if $v$ is a sum gate, and
$\Ext{v}{u}=\Prod{w}\pr \ext{v}$ if $v$ is a product gate.

A \emph{node-cut} in a circuit is a set $U$ of its nodes (gates) such
that every input-output path contains a node in $U$. Similarly, an
\emph{edge-cut} is a set $E$ of edges such that every input-output
path contains an edge in~$E$.  Recall that, in our notation, ``$f=h$''
for two polynomials $f$ and $h$ only means that their \emph{sets} of
monomials are the same---their multiplicities (coefficients) may
differ.

\begin{lem}\label{lem:cuts}
  If $U$ is a node-cut and $E$ an edge-cut in a circuit $\F$, then
  \[
  F=\sum_{u\in U}\Prod{u}\pr \ext{u}=\sum_{(u,v)\in E}\Prod{u}\pr
  \Ext{v}{u}\,.
  \]
\end{lem}

\begin{proof}
  The fact that all monomials of the last two polynomials are also
  monomials of $F$ follows from their definitions. So, it is enough to
  show that every monomial $p\in F$ belongs to both of these
  polynomials. For this, take a parse graph $\F_p$ of $p$.  Since $U$
  forms a node-cut, the graph $\F_p$ must contain some node $u\in
  U$. The monomial $p$ has a form $p=p'p''$ where $p'$ is the monomial
  produced by the subgraph of $\F_p$ rooted in $u$. Hence, $p'\in
  \Prod{u}$ and $p''\in \ext{u}$.  Similarly, since $E$ forms an
  edge-cut, the graph $\F_p$ contains some edge $(u,v)\in E$. The
  monomial $p$ has the form $p=p'p''$ where $p'$ is the monomial
  produced by the subgraph of $\F_p$ rooted in $u$. Hence, $p'\in
  \Prod{u}$ and $p''\in \Ext{v}{u}$.
\end{proof}

\section{Bounds for $(k,l)$-free Polynomials}
\label{sec:kl-free}

A polynomial $f$ is $(k,l)$-\emph{free} ($1\leq k\leq l$) if $f$ does
not contain a product of two polynomials, one with $>k$ monomials and
the other with $>l$ monomials. A polynomial $f$ is $f$-\emph{free} if
it is $(k,k)$-free, that is, if
\[
\mbox{$A\pr B\subseteq f$ implies $\min\{|A|,|B|\}\leq k$.}
\]
Note that this alone gives no upper bound on the total number $|A\pr
B|$ of monomials in the product $A\pr B$.

\begin{thm}\label{thm:GS}
  If a $(k,l)$-free polynomial $f$ can be produced by a circuit of
  size $s$, then $f$ can be written as a sum of at most $2s$ products
  $A\times B$ with $|A|\leq k$ and $|B|\leq l^2$.  In particular,
  \[
  \R[f]\geq \frac{|f|}{2kl^2}\,.
  \]
\end{thm}

\begin{proof}
  Our argument is a mix of ideas of \Cite{Gashkov and Sergeev}{GS},
  and of \Cite{Pippenger}{pippenger80}.  Take a minimal circuit $\F$
  producing $f$; hence, $F=f$ is $(k,l)$-free. This implies that every
  product gate $u=v\pr w$ in $\F$ must have an input, say $w$, at
  which a ``small'' set $A=|\Prod{w}|$ of only $|A|\leq l$ monomials
  is produced.  We thus can remove the edge $(w,u)$ and replace $u$ by
  a unary (fanin-$1$) gate $u=v\pr A$ of scalar multiplication by this
  fixed (small) polynomial $A$.  If both inputs produce small
  polynomials, then we eliminate only one of them.  What we achieve by
  doing this is that input gates remain the same as in the original
  circuit (variables $x_1,\ldots,x_n$ and constants $\nulis,\vienas$),
  each product gate has fanin $1$, and for every edge $(u,v)$ in the
  resulting circuit $\F'$, we have an upper bound
  \begin{equation}\label{eq:GS1}
    |\Ext{v}{u}|\leq l\cdot |\ext{v}|\,.
  \end{equation}
  Say that an edge $(u,v)$ in $\F'$ is \emph{legal} if both
  $|\Prod{u}|\leq k$ and $|\Ext{v}{u}|\leq l^2$ hold. Let $E$ be the
  set of all legal edges; hence, $\size{\F}\geq |E|/2$.  By
  Lemma~\ref{lem:cuts}, it remains to show that $E$ forms an edge-cut
  of $\F'$.

  To show this, take an arbitrary input-output path $P$ in $\F'$, and
  let $e=(u,v)$ be the last gate of $P$ with $|\Prod{u}|\leq k$.  If
  $v$ is the output gate, then $\ext{v}$ is a trivial polynomial
  $\vienas$, and hence, $|\Ext{v}{u}|\leq l$ by \eqref{eq:GS1},
  meaning that $(u,v)$ is a legal edge. Suppose now that $v$ is not
  the output gate.  Then $|\Prod{u}|\leq k$ but $|\Prod{v}|> k$. Held
  also $|\Ext{v}{u}|> l^2$, then \eqref{eq:GS1} would imply that
  $|\ext{v}|\geq |\Ext{v}{u}|/l > l$. Together with $|\Prod{v}|> k$
  and $\Prod{v}\pr \ext{v}\subseteq F$, this would contradict the
  $(k,l)$-freeness of $F$. Thus, $|\Prod{u}|\leq k$ and
  $|\Ext{v}{u}|\leq l^2$, meaning that $(u,v)$ is a legal edge.
\end{proof}

Together with Theorem~\ref{thm:homog}, Theorem~\ref{thm:GS} yields the
following lower bound over tropical semirings for polynomials, whose
only lower or higher envelopes are required to be $(k,l)$-free.

\begin{cor}
  Let $f$ and $g$ be polynomials such that $\lenv{f}$ and $\henv{g}$
  are $(k,l)$-free for some $1\leq k\leq l$. Then
  \[
  \Min(f)\geq \frac{|\lenv{f}|}{2kl^2}\ \ \mbox{ and }\ \ \Max(g)\geq
  \frac{|\henv{g}|}{2kl^2}\,.
  \]
\end{cor}

\begin{remark}
  Using a deeper analysis of circuit structure, \Cite{Gashkov and
    Sergeev}{gashkov,GS} were able to even estimate the numbers of sum
  and product gates: every monotone arithmetic circuit computing a
  $(k,l)$-free polynomial $f$ of $n$ variables must have at least
  $|f|/K-1$ sum gates, and at least $2\sqrt{|f|/K}-n-2$ product gates,
  where $K=\max\{k^3,l^2\}$ .
\end{remark}

\begin{remark}
  Every boolean $n\times n$ matrix $A=(a_{ij})$ defines a a set
  $Ay=(f_1,\ldots,f_n)$ of $n$ linear polynomials $f_i(y)=\sum_j
  a_{ij}y_j$, as well as a single-output bilinear polynomial $f_A(x,y)
  =\sum_{i} x_i f_i(y) =\sum_{i,j\colon a_{ij}=1} x_iy_j$ on $2n$
  variables. Call a boolean matrix $A$ $(k,l)$-\emph{free}, if it does
  not contain any $(k+1,l+1)$ all-$1$ submatrix. It is clear that the
  polynomial $f_A$ is $(k,l)$-free if and only if the matrix $A$ is
  $(k,l)$-free.

  Results of \Cite{Nechiporuk}{Necipo2} (re-discovered later by
  \Cite{Mehlhorn}{mehlhorn79} and \Cite{Pip\-peng\-er}{pippenger80})
  imply that, if $A$ is $(k,k)$-free, then $\B(Ax)\geq |A|/4k^3$,
  where $|A|$ is the number of $1$-entries in $A$. This, however, does
  not immediately yield a similar lower bound on $\B(f_A)$ for the
  \emph{single-output} version $f_A$ and, in fact, no such bound is
  known so far in the boolean semiring. (A lower bound $\B(f_A)\geq
  |A|$ for $(1,1)$-free matrices is only known when restricted to
  circuits with gates of fanout $1$; see \cite[Theorem
  7.2]{myBFC-book}.) On the other hand, Theorem~\ref{thm:GS} gives
  such a bound at least for tropical and multilinear boolean circuits:
  if $A$ is $(k,k)$-free, then
  \[
  \Min(f_A)=\Max(f_A)=\Mult{\B}{f_A}=\R[f_A]\geq |A|/2k^3\,,
  \]
  where the equalities follow from Theorem~\ref{thm:homog}, because
  the polynomial $f_A$ is homogeneous.
\end{remark}

\section{Rectangle Bound}
\label{sec:rect}

An $m$-\emph{balanced product-polynomial} is a product of two
polynomials, one of which has minimum degree $d$ satisfying $m/3<
d\leq 2m/3$, and is itself a product of two nonempty polynomials.

\begin{lem}[Sum-of-Products]\label{lem:sum-of-products}
  If a polynomial $f$ of minimum degree at least $m\geq 3$ can be
  produced by a circuit with $s$ product gates, then $f$ can be
  written as a sum of at most $s$ $m$-balanced product-polynomials.
\end{lem}

\begin{proof}
  Let $d$ be the minimum degree of $f$, and $\F$ be a circuit with $s$
  product gates producing $f$.  Hence, $F=f$ and $d\geq m$.  By the
  degree $\dg{\gate}$ of a gate gate $\gate\in\F$ we will mean the
  minimum degree of the polynomial produced at $\gate$. In particular,
  the degree of the output gate is~$d$.

  \begin{clm}\label{clm:moderate}
    For every $\epsilon\in(1/d,1)$, there exists a product gate
    $\gate$ with $\dg{\gate}\in (\epsilon d/2, \epsilon d]$.
  \end{clm}

\begin{proof}
  Start at the output gate of $\F$, and traverse the circuit (in the
  reverse order of edges) by always choosing the input of larger
  degree until a gate $v=u\ast v$ of degree $\dg{v}> \epsilon d$ is
  found such that both $\dg{u}$ and $\dg{w}$ are $\leq \epsilon d$.
  Assume w.l.o.g.  that $\dg{u}\geq \dg{w}$.  Since $\dg{v}\leq \dg{u}
  + \dg{w}\leq 2 \dg{u}$, the gate $u$ has the desired degree
  $\epsilon d/2 < \dg{u}\leq \epsilon d$.  If the gate $u$ is a sum
  gate, then at least one of its inputs must have the same degree
  $\dg{u}$. So, we can traverse the circuit further until a product
  gate of degree $\dg{u}$ is found.
\end{proof}

Now, we apply Claim~\ref{clm:moderate} with $\epsilon:=2m/3d$ to find
a product gate $\gate$ of degree $m/3=\epsilon d/2\leq
\dg{\gate}\leq\epsilon d=2m/3$.  By Lemma~\ref{lem:contain}, we can
write $F$ as $F=F_{\gate} \su F_{\gate=\nulis}$ where $F_{\gate}=A\pr
B$ is a product of two polynomials such that the minimum degree of $A$
lies between $m/3$ and $2m/3$, and $A$ itself is a product of two
nonempty polynomials (since $\gate$ is a product gate); hence,
$F_{\gate}$ is an $m$-balanced product-polynomial.  The polynomial
$F_{\gate=\nulis}$ is obtained from $F$ by removing some monomials. If
$F_{\gate=\nulis}$ is empty, then we are done.  Otherwise, the
polynomial $F_{\gate=\nulis}$ still has minimum degree at least $m$,
and can be produced by a circuit with one product gate fewer.  So, we
can repeat the same argument for it, until the empty polynomial is
obtained.
\end{proof}

\begin{remark}
  Lemma~\ref{lem:sum-of-products} remains true if, instead of the
  minimum degree measure $\deg{f}$ of polynomials, one takes the
  minimum length $\length{f}$ of a monomial of $f$, where the
  \emph{length} of a monomial $p$ is defined as the number $|X_p|$ of
  distinct variables occurring in $p$. Hence, we always have that
  $\deg{f}\geq \length{f}$, and $\deg{f}=\length{f}$ holds if $f$ is
  multilinear.  The same argument works because
  $\length{F_{\gate=\nulis}}\geq \length{F}$, as long as the
  polynomial $F_{\gate=\nulis}$ is not empty.
\end{remark}
To upper-bound the maximal possible number $|A\pr B|$ of monomials in
a product-polynomial $A\pr B\subseteq f$, the following measure of
\emph{factor-density} naturally arises: for an integer $r\geq 0$, let
$\ddeg{f}{r}$ be the maximum number of monomials in $f$ containing a
fixed monomial of degree~$r$ as a common factor.  This measure tells
us how much the monomials of $f$ are ``stretched'': the faster
$\ddeg{f}{r}$ decreases with increasing $r$, the more stretched $f$
is.  Note that, if $d$ is the maximum degree of $f$, then
\[
1=\ddeg{f}{d}\leq \ddeg{f}{d-1}\leq \ldots\leq \ddeg{f}{1} \leq
\ddeg{f}{0}=|f|\,.
\]

The factor-density measure allows to upper-bound the number of
monomials in pro\-duct-po\-ly\-no\-mials over any semiring which is
not multiplicatively-idempotent (where $a^2=a$ holds only for
$a=\vienas$).  Such are, in particular, the arithmetic semiring as
well as all four tropical semirings. The only property of such
semirings we will use is that, if $p$ is a monomial and $A$ is a
polynomial, then $|A|\leq |\{p\}\pr A|$ holds.  Note that this needs
not to hold in semirings which \emph{are} multiplicatively-idempotent:
the polynomial $A=\{x,y\}$ has two monomials, but $\{xy\}\pr A=\{x^2y,
xy^2\}\ =\{xy\}$ has only one monomial.

\begin{obs}\label{obs:product}
  Let $A$ and $B$ be polynomials over a not
  multiplicatively-idempotent semiring of maximum degrees $a$ and $b$.
  If $A\pr B\subseteq f$, then $|A\pr B|\leq \ddeg{f}{a}\cdot
  \ddeg{f}{b}$.
\end{obs}

\begin{proof}
  Fix a monomial $p\in A$ of degree $|p|=a$, and a monomial $q\in B$
  of degree $|q|=b$.  Since $\{p\}\pr B\subseteq f$, we have that
  $|B|\leq |\{p\}\pr B|\leq \ddeg{f}{|p|}= \ddeg{f}{a}$.  Similarly,
  since $A\pr\{q\}\subseteq f$, we have that $|A|\leq |A\pr \{q\}|\leq
  \ddeg{f}{|q|} = \ddeg{f}{b}$.
\end{proof}

\begin{lem}[Rectangle Bound]\label{lem:simple}
  For every polynomial $f$ of minimum degree at least $m\geq 3$, there
  is an integer $m/3< r\leq 2m/3$ such that
  \[
  \R[f]\geq \frac{|f|}{\ddeg{f}{r}\cdot \ddeg{f}{m-r}}\,.
  \]
  Moreover, the lower bound is on the number of product gates.
\end{lem}

\begin{proof}
  Let $\F$ be a minimal monotone arithmetic circuit representing $f$,
  and let $s=\size{\F}$. By Lemma~\ref{lem:sum-of-products}, the
  polynomial $F=f$ can be written as a sum of at most $s$ products
  $A\pr B$ of polynomials, where the minimum degree $a=\deg{A}$ of $A$
  satisfies $m/3\leq a\leq 2m/3$; hence, $\deg{B}\geq
  m-a$. Observation~\ref{obs:product} implies that $|A\pr B| \leq
  \ddeg{f}{\deg{A}}\cdot\ddeg{f}{\deg{B}} \leq \ddeg{f}{a}\cdot
  \ddeg{f}{m-a}$.
\end{proof}

The Rectangle Bound allows one to easily obtain strong lower bounds
for some explicit polynomials.

\begin{thm}\label{thm:bounds1}
  If $f\in\{\PERM{n}, \HC{n}, \ST{n}\}$, then $\Min(f)$, $\Max(f)$ and
  $\Mult{\B}{f}$ are $2^{\Omega(n)}$.
\end{thm}

\begin{proof}
  Since all these three polynomials $f$ are multilinear and
  homogeneous, it is enough (by Theorem~\ref{thm:homog}) to prove the
  corresponding lower bounds on $\R[f]$. We will obtain such bounds by
  applying Lemma~\ref{lem:simple}.

  The permanent polynomial $f=\PERM{n}$ has $|f|=n!$ multilinear
  monomials $x_{1,\pi(1)}x_{2,\pi(2)}\cdots x_{n,\pi(n)}$, one for
  each permutation $\pi:[n]\to[n]$. Since at most $(n-r)!$ of the
  permutations can take $r$ pre-described values, we have that
  $\ddeg{f}{r}\leq (n-r)!$. (In fact, here we even have the equality
  $\ddeg{f}{r}= (n-r)!$.)  Lemma~\ref{lem:simple} gives $\R[f]\geq
  n!/(n-r)!r! = \binom{n}{r}$ for some $n/3<r\leq 2n/3$; so,
  $\R[f]=2^{\Omega(n)}$.

  The argument for $\HC{n}$ is almost the same: the only difference is
  that now the monomials correspond to symmetric, not to all
  permutations.

  The spanning tree polynomial $f=\ST{n}$ is a homogeneous polynomial
  of degree $n-1$ with $|f|=n^{n-2}$ monomials
  $x_{2,\pi(2)}x_{3,\pi(3)}\cdots x_{n,\pi(n)}$ corresponding to the
  functions $\pi:\{2,3,\ldots,n\}\to [n]$ such that $\forall i$
  $\exists k$: $\pi^{(k)}(i)=1$. Each spanning tree gives a function
  with this property by mapping sons to their father.  Now, if we fix
  some $r$ edges, then $r$ values of functions $\pi$ whose spanning
  trees contain these edges are fixed. Thus, $\ddeg{f}{r}\leq
  (n-r)^{n-r-2}$, and Lemma~\ref{lem:simple} gives
  $\R[f]=2^{\Omega(n)}$.
\end{proof}

Using a tighter analysis (in the spirit of
Remark~\ref{rem:double-counting}) and more involved computations,
\Cite{Jerrum and Snir}{jerrum} obtained even \emph{tight} lower bounds
for $\PERM{n}$ and $\HC{n}$.

The three polynomials in Theorem~\ref{thm:bounds1} are homogeneous. To
show that the rectangle bound works also for non-homogeneous
polynomials, consider the $st$-con\-nec\-ti\-vi\-ty polynomial
$\STCONN{n}$.  We know that this polynomial has $\Min$-circuits of
size $O(n^3)$ (Remark~\ref{rem:ST}). But $\Max$-circuits for this
polynomial must be of exponential size.

\begin{thm}\label{thm:bounds2}
  If $f=\STCONN{n+2}$, then $\Max(f)$ and $\Min[f]$ are at least
  $2^{\Omega(n)}$.
\end{thm}

\begin{proof}
  Consider the higher envelope $\henv{f}$ of $f$.  This is a
  homogeneous polynomial of degree $n$ with $|\henv{f}|=n!$ monomials
  corresponding to paths in $K_{n+2}$ from $s=0$ to $t=n+1$ with
  exactly $n$ inner nodes. Since $\ddeg{f}{r}\leq (n-r)!$,
  Lemma~\ref{lem:simple} (with $r=n/3$) gives
  $\R[\henv{f}]=2^{\Omega(n)}$.  By Theorem~\ref{thm:homog}, the same
  lower bound holds for $\Max(f)$ and $\Min[f]$.
\end{proof}

\section{Truly Exponential Lower Bounds}
\label{sec:trully}

Note that the lower bounds above have the forms
$2^{\Omega(\sqrt{n})}$, where $n$ is the number of variables. Truly
exponential lower bounds $\R[f]=\Omega(2^{n/2})$ on the monotone
circuit size of multilinear polynomials of $n$ variables were
announced by \Cite{Kasim-Zade}{oktai1,oktai2}. Somewhat earlier, a
lower bound $\R[f]=2^{\Omega(n)}$ was announced by
\Cite{Kuznetsov}{kuznetsov}.  Then, \Cite{Gashkov}{gashkov} proposed a
general lower bounds argument for monotone arithmetic circuits and
used it to prove an $\R[f]=\Omega(2^{2n/3})$ lower bound.

The construction of the corresponding multilinear polynomials in these
works is algebraic. Say, the monomials of the polynomial $f(x,y)$ of
$2n$ variables constructed in \cite{oktai1,oktai2} have the form
$x_1^{a_1}\cdots x_n^{a_n}y_1^{b_1}\cdots y_n^{b_n}$ where $a\in
\gf{2}^n$ and $b = a^3$ (we view vector $a$ as an element of
$\gf{2^n}$ when rising it to the 3rd power).  That is, monomials
correspond to the points of the cubic parabola $\{(a,a^3)\colon a\in
\gf{2^n}\}$.  The monomials of the polynomial constructed in
\cite{gashkov} are defined using triples $(a,b,c)$ with $a,b,c\in
\gf{2^n}$ satisfying $a^3+b^7+c^{15}=1$.  The constructed polynomials
are $(k,l)$-free for particular \emph{constants} $k$ and $l$, and the
desired lower bounds follow from general lower bounds of
\Cite{Gashkov}{gashkov}, and \Cite{Gashkov and Sergeev}{GS} for
$(k,l)$-free polynomials (see Sect.~\ref{sec:kl-free} for these
bounds).

Without knowing these results, \Cite{Raz and Yehudayoff}{RY} have
recently used discrepancy arguments and exponential sum estimates to
derive a truly exponential lower bound $\R[f]=2^{\Omega(n)}$ for an
explicit multilinear polynomial $f(x_1,\ldots,x_n)$. Roughly, their
construction of $f$ is as follows. Assume that $n$ divided by a
particular constant $k$ is a prime number. View a monomial $p$ as a
$0/1$ vector of its exponents. Split this vector into $k$ blocks of
length $n/k$, view each block as a field element, multiply these
elements, and let $c_p\in\{0,1\}$ be the first bit of this
product. Then include the monomial $p$ in $f$ if and only if $c_p=1$.

In this section we use some ideas from \cite{juk2008} to show that
truly exponential lower bounds can be also proved using graphs with
good expansion properties. Numerically, our bounds (like those in
\cite{RY}) are worse than the bounds in
\cite{oktai1,oktai2,gashkov,GS} (have smaller constants), but the
construction of polynomials is quite simple (modulo the construction
of expander graphs).

Say that a partition $[n]=S\cup T$ is \emph{balanced} if $n/3\leq
|S|\leq 2n/3$.  Define the \emph{matching number} $\match{G}$ of a
graph $G=([n],E)$ as the largest number $m$ such that, for every
balanced partition of nodes of $G$, at least $m$ crossing edges form
an induced matching.  An edge is crossing if it joins a node in one
part of the partition with a node in the other part.  Being an induced
matching means that no two endpoints of any two edges of the matching
are joined by a crossing edge.

Our construction of hard polynomials is based on the following lemma.
Associate with every graph $G=([n],E)$ the multilinear polynomial
$f_G(x_1,\ldots,x_n)$ whose monomials are $\prod_{i\in S}x_i$ over all
subsets $S\subseteq [n]$ such that the induced subgraph $G[S]$ has an
odd number of edges of~$G$.

\begin{lem}\label{lem:main-exp}
  For every non-empty graph $G$ on $n$ nodes, we have
  \[
  \R[f_G]\geq 2^{\match{G}-2}\,.
  \]
\end{lem}

We postpone the proof of this lemma and turn to its application.

The following simple claim gives us a general lower bound on the
matching number $\match{G}$. Say that a graph is $s$-\emph{mixed} if
every two disjoint $s$-element subsets of its nodes are joined by at
least one edge.

\begin{clm}
  If an $n$-node graph $G$ of maximum degree $d$ is $s$-mixed, then
  $\match{G}\geq (\lfloor n/3\rfloor-s)/(2d+1)$.
\end{clm}

\begin{proof}
  Fix an arbitrary balanced partition of the nodes of $G$ into two
  parts.  To construct the desired induced matching, formed by
  crossing edges, we repeatedly take a crossing edge and remove it
  together with all its neighbors. At each step we remove at most
  $2d+1$ nodes.  If the graph is $s$-mixed, then the procedure will
  run for $m$ steps as long as $\lfloor n/3\rfloor-(2d+1)m$ is at
  least~$s$.
\end{proof}

Thus, we need graphs of small degree that are still $s$-mixed for
small $s$.  Examples of such graphs are expander graphs.  A
\emph{Ramanujan graph} is a regular graph $G_{n,q}$ of degree $q+1$ on
$n$ nodes such that $\lambda(G)\leq 2\sqrt{q}$, where $\lambda(G)$ is
the second largest (in absolute value) eigenvalue of the adjacency
matrix of $G$.  Explicit constructions of Ramanujan graphs on $n$
nodes for every prime $q\equiv 1\bmod{4}$ and infinitely many values
of $n$ were given by \Cite{Margulis}{Margulis}, \Cite{Lubotzky,
  Phillips and Sarnak}{LPS}; these were later extended to the case
where $q$ is an arbitrary prime power by
\Cite{Morgenstern}{Morgenstern}, and \Cite{Jordan and Livn\'e}{JL97}.

\begin{thm}\label{thm:trully}
  If $f_G(x_1,\ldots,x_n)$ is the multilinear polynomial associated
  with the Ramanujan graph $G=G_{n,64}$, then
  \[
  \R[f_G]\geq 2^{0.001 n}\,.
  \]
\end{thm}

\begin{proof}
  The Expander Mixing Lemma (\cite[Lemma~2.3]{AC88}) implies that, if
  $G$ is a $d$-regular graph on $n$ nodes, and if $s >\lambda(G)\cdot
  n/d$, then $G$ is $s$-mixed.  Now, the graph $G=G_{n,q}$ is
  $d$-regular with $d=q+1$ and has $\lambda(G)\leq 2\sqrt{q}$.  Hence,
  the graph $G$ is $s$-mixed for $s=2n/\sqrt{q}> 2\sqrt{q}n/(q+1)$.

  Our graph $G=G_{n,64}$ is a regular graph of degree $d=65$, and is
  $s$-mixed for $s=2n/\sqrt{64}=n/4$. Lemma~\ref{lem:main-exp} gives
  the desired lower bound.
\end{proof}

It remains to prove Lemma~\ref{lem:main-exp}.

Call polynomial $f(x_1,\ldots,x_n)$ a \emph{product polynomial}, if
$f$ is a product of two polynomials on disjoint sets of variables,
each of size at least~$n/3$, that is, if $f=g(Y)\pr h(Z)$ for some
partition $Y\cup Z=\{x_1,\ldots,x_n\}$ of variables with $|Y|,|Z|\geq
n/3$, and some two polynomials $g$ and $h$ on these variables.  Note
that we do not require that, say, the polynomial $g(Y)$ must depend on
all variables in $Y$: some of them may have zero degrees in~$g$.

\begin{clm}[\cite{RY}]\label{clm:struct2}
  If $\F(x_1,\ldots,x_n)$ is a multilinear circuit of size $s$ with
  $n\geq 3$ input variables, then the polynomial $F$ can be written as
  a sum of at most $s+1$ product polynomials.
\end{clm}

\begin{proof}
  Induction on $s$. For a gate $\gate$, let $X_{\gate}$ be the set of
  variables in the corresponding subcircuit of $\F$. Let $v$ be the
  output gate of $\F$. If $v$ is an input gate, then $F$ itself is a
  product polynomial, since $n\geq 3$. So, assume that $v$ is not an
  input gate.  If $|X_v|\leq 2n/3$, then the polynomial $F$ itself is
  a product polynomial, because $F=F\pr \vienas$. So, assume that
  $|X_v|> 2n/3$.  Every gate $u$ in $\F$ entered by gates $u_1$ and
  $u_2$ admits $|X_u|\leq |X_{u_1}|+|X_{u_2}|$.  Thus, there exists a
  gate $\gate$ in $\F$ such that $n/3\leq |X_{\gate}|\leq 2n/3$.  By
  Lemma~\ref{lem:contain}, we can write $F$ as $F=F_{\gate} \su
  F_{\gate=\nulis}$ where $F_{\gate}=g_{\gate}\pr h$ with $n/3\leq
  |X_{\gate}|\leq 2n/3$ and some polynomial $h$. Moreover, since the
  circuit is multilinear, the set $X_h$ of variables in the polynomial
  $h$ must be disjoint from $X_{\gate}$, implying that $|X_h|\geq
  n-|X_{\gate}| \geq n/3$.  Thus, $g_u\pr h$ is a product polynomial.
  Since the circuit $F_{\gate=\nulis}$ has size at most $s-1$, the
  desired decomposition of $F$ follows from the induction hypothesis.
\end{proof}

By the \emph{characteristic function} of a multilinear polynomial
$f(x_1,\ldots,x_n)$ we will mean the (unique) boolean function which
accepts a binary vector $a\in\{0,1\}^n$ if and only if the polynomial
$f$ contains the monomial $x_1^{a_1}x_2^{a_2}\cdots
x_n^{a_n}=\prod_{i\colon a_i=1}x_i$.  (Note that this boolean function
needs not to be monotone.)  In particular, the characteristic function
of our polynomial $f_G$ is the quadratic boolean function
\[
\phi(x) = \sum_{\{i,j\}\in E} x_i x_j\bmod{2}\,.
\]
That is, $\phi(a)=1$ if the subgraph $G[S]$ induced by the set of
nodes $S=\{i\colon a_i=1\}$ has an odd number of edges.  Since
$\phi(x)$ is a non-zero polynomial of degree $2$ over $\gf{2}$, we
have that $|f_G|=|\phi^{-1}(1)|\geq 2^{n-2}$.

\begin{clm}\label{clm:matching}
  For every graph $G$ on $n$ nodes, every product sub-polynomial of
  $f_G$ contains at most $2^{n-\match{G}}$ monomials.
\end{clm}

\begin{proof}
  Let $G\pr H$ be a product polynomial contained in $f_G$. This
  polynomial gives a partition $x=(y,z)$ of the variables into two
  parts, each containing at least $n/3$ variables. Let $g(y)$ and
  $h(z)$ be the characteristic functions of $G$ and $H$, and
  $r(x)=g(y)\land h(z)$.  Then $|G\pr H|=|r^{-1}(1)|$, and it is
  enough to show that $|r^{-1}(1)|\leq 2^{n-\match{G}}$. When doing
  this, we will essentially use the fact that $r\leq \phi$, which
  follows from the fact that all monomials of $G\pr H$ are also
  monomials of~$f_G$.

  By the definition of $m(G)$, some set $M=\{y_1 z_1,\ldots,y_m z_m\}$
  of $m= m(G)$ crossing edges $y_iz_i$ forms an induced matching of
  $G$.  Given an assignment $\alpha$ of constants $0$ and $1$ to the
  $n-2m$ variables outside the matching $M$, define vectors $a,b\in
  \{0,1\}^m$ and a constant $c\in\{0,1\}$ as follows:
  \begin{itemize}
  \item $a_i=1$ iff an odd number of neighbors of $y_i$ get value $1$
    under $\alpha$,
  \item $b_i=1$ iff an odd number of neighbors of $z_i$ get value $1$
    under $\alpha$,
  \item $c=1$ iff the number of edges whose both endpoints get value
    $1$ under $\alpha$ is odd.
  \end{itemize}
  Then the subfunction $\phi_{\alpha}$ of $\phi$ obtained after
  restriction $\alpha$ is
  \begin{align*}
    \phi_{\alpha}(y_1,\ldots,y_m,z_1,\ldots,z_m)&=\sum_{i=1}^m y_i z_i + \sum_{i=1}^m y_ia_i + \sum_{i=1}^m b_iz_i + c \mod{2}\\
    &=IP_m(y\oplus b,z\oplus a)\oplus IP_m(a,b)\oplus c\,,
  \end{align*}
  where $IP_n(y_1,\ldots,y_m,z_1,\ldots,z_m)=\sum_{i=1}^m
  y_iz_i\mod{2}$ is the inner product function (sca\-lar product).
  Since $a,b$ and $c$ are \emph{fixed}, the corresponding $2^m\times
  2^m$ $\pm 1$ matrix $H$ with entries
  $H[y,z]=(-1)^{\phi_{\alpha}(y,z)}$ is a Hadamard matrix (rows are
  orthogonal to each other).  Lindsey's Lemma (see,
  e.g. \cite[p.~479]{myBFC-book}) implies that no monochromatic
  submatrix of $H$ can have more than $2^m$ $1$-entries.

  Now, the obtained subfunction $r_{\alpha}=g_{\alpha}(y_1,\ldots,y_m)
  \land h_{\alpha}(z_1,\ldots,z_m)$ of $r=g(y)\land h(z)$ also
  satisfies $r_{\alpha}(a,b)\leq \phi_{\alpha}(a,b)$ for all
  $a,b\in\{0,1\}^{m}$. Since the set of all pairs $(a,b)$ for which
  $r_{\alpha}(a,b)=1$ forms a \emph{submatrix} of $H$, this implies
  that $r_{\alpha}$ can accept at most $2^m$ such pairs.  Since this
  holds for each of the $2^{n-2m}$ assignments $\alpha$, the desired
  upper bound $|r^{-1}(1)|\leq 2^m\cdot 2^{n-2m}=2^{n-m}$ follows.

  This completes the proof of Claim~\ref{clm:matching}, and hence, the
  proof of Lemma~\ref{lem:main-exp}.
\end{proof}

\section{Depth Lower Bounds}

So far, we were interested in the \emph{size} of circuits. Another
important measure is the circuit \emph{depth}, i.e. the number of
nodes in a longest input-output path.  For a polynomial $f$, let
$\depth{f}$ denote the smallest possible depth of a circuit
producing~$f$.

If a polynomial $f$ can be produced by a circuit of \emph{size} $s$,
what is then the smallest \emph{depth} of a circuit producing~$f$?
\Cite{Hyafil}{hyafil} has shown that then $f$ can be also produced by
a circuit of depth proportional to $(\log d)(\log sd)$, where $d$ is
the maximum degree of~$f$. (This can be easily shown by induction on
the degree using the decomposition given in
Lemma~\ref{lem:sum-of-products}.) However, the size of the resulting
circuit may be as large as $s^{\log d}$.  A better simulation, leaving
the size polynomial in $s$, was found by \Cite{Valiant et al.}{VSBC}.

\begin{thm}[\Cite{Valiant et al.}{VSBC}]
  If a polynomial $f$ of maximum degree $d$ can be produced by a
  circuit of size $s$, then $f$ can be also produced by a circuit of
  size $O(s^3)$ and depth $O(\log s\log d)$.
\end{thm}
In particular, if a multilinear polynomial $f$ of $n$ variables can be
\emph{produced} by a circuit $\F$ of polynomial in $n$ size, then
$\depth{f}=O(\log^2 n)$.  By Lemma~\ref{lem:multilin},
$\depth{f}=O(\log^2 n)$ also holds if $f$ is only \emph{computed} by a
$\Max$, $\M$ or $\MAX$ circuit of polynomial size.  This, however, no
more holds for $\B$ and $\Min$ circuits: even though $\pol{F}=\pol{f}$
holds over these semirings, the produced polynomial $F$ may have
maximum degree exponential in~$n$.

We now turn to proving lower bounds on $\depth{f}$.  In the previous
section, we have shown that the factor-density measure $\ddeg{f}{r}$
can be used to lower bound the circuit size.  By simplifying previous
arguments of \Cite{Shamir and Snir}{shamir}, \Cite{Tiwari and
  Tompa}{tiwari} have shown that the measure $\ddeg{f}{r}$ can be also
used to lower bound the circuit depth as well. The idea was
demonstrated in \cite{tiwari} on two applications
(Theorem~\ref{thm:TT1} and \ref{thm:TT2} below). Here we put their
idea in a general frame.

A \emph{subadditive weighting} of a circuit $\F$ is an assignment
$\wweight:\F\to \RR_+$ of non-negative weights to the gates of $\F$
such that the output gate gets weight $\geq 1$, all other gates get
weight $\leq 1$, and and $\weight{v\su w}\leq \weight{v}+\weight{w}$
holds for every sum gate $v\su w$.  Given such a weighting, define the
\emph{decrease} $\factor{u}$ at a product gate $u=v\pr w$ as
\[
\factor{u}=\frac{\weight{v}\cdot \weight{w}}{\weight{u}}\,.
\]
Note that, since $\weight{v}\leq 1$ holds for every non-output gate
$v$, we have
\[
\weight{u}\leq \frac{1}{\factor{u}}\cdot\min\{\weight{v},
\weight{w}\}\,.
\]
That is, when entering $u$ from any of its two inputs, the weight must
decrease by a factor of at least $\factor{u}$. This explains the use
of term ``decrease''.  Let $\factor{r,s}=\min_u \factor{u}$ be the
smallest decrease at a product gate $u$ of degree~$r$, one of whose
inputs has degree~$s$; by the \emph{degree} of a gate we mean the
minimum degree of the polynomial produced at that gate.

\begin{lem}
  Let $\F$ be a circuit, whose produced polynomial has minimum degree
  $d$, and let $m=\log_2 d$.  Then, for every subadditive weighting,
  there is sequence $d=r_0> r_1 > \ldots > r_m =1$ of integers such
  that $r_{i+1}\geq \frac{1}{2} r_i$ for all $i=1,\ldots,m$, and the
  circuit $\F$ has depth at least
  \[
  m+\log_2 \prod_{i=0}^{m-1} \factor{r_i,r_{i+1}}\,.
  \]
\end{lem}

\begin{proof}
  Construct a path $\pi$ from the output gate to an input gates as
  follows: at a sum gate choose the input of greater weight, and at a
  product gate choose an input of greater degree. Since the produced
  polynomial has minimum degree $d$, and since at each product gate we
  chose an input of greater degree, there must be at least $m$ product
  gates along $\pi$.  Let $d=r_1> r_2 > \ldots > r_m > r_{m+1}=1$ be
  the degrees of the product gates (and input node) on path $\pi$.
  Let $k_i=\factor{r_i,r_{i+1}}$ be the decrease of the $i$-th product
  gate.  Note by the construction of $\pi$ that $r_{i+1}\geq
  \frac{1}{2} r_i$.

  Let us now view the path $\pi$ in the reversed order (from input to
  output).  So, we start with some gate of weight $\leq 1$ (an input
  gate).  Since the weighting is subadditive, at each edge entering a
  sum gate the weight can only increase by a factor of at most
  $2$. So, if $s$ is the number of sum gates along~$\pi$, then the
  total increase in weight is by a factor at most $2^s$. But when
  entering the $i$-th product gate, the weight decreases by a factor
  at least $k_i$.  Thus, the total loss in the weight is by a factor
  at least $\prod_{i=0}^{m-1} k_i$. Since the last (output) gate must
  have weight $\geq 1$, this gives
  \[
  2^s\cdot \prod_{i=0}^{m-1} \frac{1}{k_i}\geq 1\,, \ \mbox{ and
    hence, }\ s\geq \log_2 \prod_{i=0}^{m-1} k_i\,.
  \]
  Since $\depth{f}\geq m+s$, we are done.
\end{proof}

We now give a specific weighting, based on the the factor-density
measure $\ddeg{f}{r}$. Recall that $\ddeg{f}{r}$ is the maximum number
of monomials in $f$ containing a fixed monomial of degree~$r$ as a
common factor. For a polynomial $f$ of minimum degree $d$, and an
integer $1\leq s<r\leq d$, define
\[
\Factor{f}{r,s}=\frac{\ddeg{f}{d-r}}{\ddeg{f}{d-s}\cdot
  \ddeg{f}{d-r+s}}\,.
\]
Note that we have already used this measure to lower-bound the size of
circuits: if $f$ is homogeneous of degree $d$, then
Lemma~\ref{lem:simple} yields $\R[f]\geq \Factor{f}{d,s}$ for some
$d/3\leq s\leq 2d/3$.

\begin{lem}\label{lem:shamir}
  Let $f$ be a polynomial of minimum degree $d$, and $m=\log_2
  d$. Then there is a sequence $d=r_0> r_1 > \ldots > r_{m}=1$ of
  integers such that $r_{i+1}\geq \frac{1}{2} r_i$ for all
  $i=1,\ldots,m$, and
  \[
  \depth{f}\geq m+\log_2\prod_{i=0}^{m-1}\Factor{f}{r_i,r_{i+1}}\,.
  \]
\end{lem}

\begin{proof}
  Let $\F$ be a circuit producing $f$; hence, $F=f$.  For a gate
  $u\in\F$, let $d_u$ be the minimum degree of the polynomial produced
  at $u$. By Theorem~\ref{lem:contain}, we know that $F$ can be
  written as a sum $F=A_u\pr B + F_{u=0}$, where $A_u$ is the
  polynomial produced at gate $u$. Since $A_u\pr B\subseteq f$, and
  $A_u$ has minimum degree $d_u$, the polynomial $B$ must contain a
  monomial $p$ of degree $|p|\geq d-d_u$.  Hence, by
  Observation~\ref{obs:product}, we have that $|A_u|\leq
  \ddeg{f}{d-d_u}$. This suggests the following weighting of gates:
  \[
  \weight{u}=\frac{|A_u|}{\ddeg{f}{d-d_u}}\,.
  \]
  The output gate $v$ then gets weight $\weight{v}\geq
  |f|/\ddeg{f}{d-d} =1$, whereas all other gates get weights $\leq 1$.
  Moreover, since for every product gate $u=v\pr w$, we have that
  $|A_u|=|A_v|\cdot|A_w|$ and $d_u=d_v+d_w$, the decrease
  $\factor{r,s}$ of this weighting coincides with
  $\Factor{f}{r,s}$. So, it remains to show that the weighting is
  subadditive.

  To show this, let $u=v+w$ be a sum gate.  Then $d_u= \min\{d_v,
  d_w\}$, and hence, $d-d_u= \max\{d-d_v, d-d_w\}$.  So,
  \begin{align*}
    \weight{v\su w}&=\frac{|A_v|+|A_w|}{\ddeg{f}{d-d_u}} =
    \frac{|A_v|+|A_w|}{\max\{\ddeg{f}{d-d_v}, \ddeg{f}{d-d_w}\}} \leq
    \weight{v}+\weight{w}\,.
  \end{align*}
\end{proof}

\begin{thm}[\cite{shamir,tiwari}]\label{thm:TT1}
  If $f=\PERM{n}$, then $\depth{f}\geq n + \log_2n-1$.
\end{thm}

\begin{proof}
  The permanent polynomial $f=\PERM{n}$ is a homogeneous multilinear
  polynomial of degree $d=n$. Moreover, $\ddeg{f}{l}=(n-l)!$ holds for
  any $1\leq l\leq d$.  Hence,
  \[
  \Factor{f}{r,s} =\frac{r!}{s!(r-s)!}  =\binom{r}{s}\,.
  \]
  But $r_{i+1}\geq \frac{1}{2}r_i$ implies that $\binom{r_i}{r_{i+1}}
  \geq 2^{r_i-r_{i+1}}$. Hence,
  \[
  \prod_{i=0}^{m-1}\Factor{f}{r_i,r_{i+1}}=
  \prod_{i=0}^{m-1}\binom{r_i}{r_{i+1}} \geq 2^{r_0-r_m}=2^{n-1}\,.
  \]
\end{proof}

This lower bound for $f=\PERM{}$ is not surprising, since $\depth{f}$
is always at least logarithmic in $\R[f]$, and we already know
(Theorem~\ref{thm:bounds1}) that $\R[f]$ is exponential for this
polynomial.  More interesting, however, is that the argument above
allows to prove super-logarithmic depth lower bounds even for
polynomials that \emph{have} circuits of polynomial size.

To demonstrate this, consider the following \emph{layered
  $st$-connectivity polynomial} $\SSTCONN{n,d}$. The monomials of this
polynomial correspond to $st$-paths in a layered graph. We have $d+1$
disjoint layers, where the first contains only one node $s$, the last
only one node $t$, and each of the remaining $d-1$ layers contains $n$
nodes.  Monomials of $\SSTCONN{n,d}$ have the form $
x_{s,a_1}x_{a_1,a_2}\cdots x_{a_{d-2},a_{d-1}}x_{a_{d-1},t}$ with
$a_i$ belonging to the $i$-th layer. In other words, this polynomial
corresponds to computing the $(s,t)$-entry of the product of $d-1$
matrices of dimension $n\times n$. Hence, it can be produced by a
circuit of depth $O((\log d)(\log n))$.

\begin{thm}[\cite{shamir,tiwari}]\label{thm:TT2}
  $\depth{\SSTCONN{n,d}}\geq (\log_2 d)(1+\log_2 n)$.
\end{thm}

\begin{proof}
  The polynomial $f=\SSTCONN{n,d}$ is a multilinear homogeneous
  polynomial of degree $d$ with $|f|=n^{d-1}$ monomials. To estimate
  the factor-density $\ddeg{f}{l}$, let us fix a set $E$ of $|E|=l$
  edges. Every edge $e\in E$ constrains either two inner nodes (if
  $s,t\not\in e$) or one inner node.  Thus, if we fix $l$ edges, then
  at least $l$ inner nodes are constrained, implying that only
  $\ddeg{f}{l}\leq n^{d-1-l}$ paths can contain all these edges. In
  fact, we have an equality $\ddeg{f}{l}= n^{d-1-l}$: every monomial
  $x_{s,a_1}x_{a_1,a_2}\cdots x_{a_{l-1},a_{l}}$ consisting of initial
  $l$ edges is a factor of exactly $n^{d-1-l}$ monomials of $f$. Thus,
  the decrease in this case is
  \[
  \Factor{f}{r,s}=\frac{\ddeg{f}{d-r}}{\ddeg{f}{d-s}\cdot
    \ddeg{f}{d-(r-s)}} =\frac{n^{r-1}}{n^{s-1}\cdot n^{r-s-1}}=n
  \]
  for all $1\leq s<r\leq d$. Lemma~\ref{lem:shamir} yields
  $\depth{f}\geq \log_2 d + \log_2 n^{\log_2 d}$, as desired.
\end{proof}

\begin{table}
  \begin{center}
    \begin{tabular*}{0.85\textwidth}{l@{\hskip 1cm}l@{\hskip 1cm}l}
      {\bf Bound} & {\bf Property of $f$} & {\bf Ref.}\\
      \noalign{\smallskip}
      $\B(f) > t$ & $f$ is not $t$-simple (Def.~\ref{def:simple})  & Thm.~\ref{thm:crit}
      \\[0.5ex]
      $\A(f)=\R[f]$ & $f$ is homogeneous & Thm.~\ref{thm:homog}\\[0.5ex]
      $\R[f]\geq |f|$ &  $f$ is separated (Def.~\ref{def:separated})& Thm.~\ref{thm:schnorr}\\[0.5ex]
      $\displaystyle \R[f]\geq \frac{|f|}{2kl^2}$ &  $A\pr B\subseteq f$ implies $|A|\leq l$ or $|B|\leq k$ & Thm.~\ref{thm:GS}\\[0.5ex]
      $\displaystyle \R[f]\geq \frac{|f|}{\ddeg{f}{r}\cdot \ddeg{f}{d-r}}$ & $f$ of minimum degree $d$ & Lem.~\ref{lem:simple}\\
      \noalign{\smallskip}
    \end{tabular*}
    \caption[]{A summary of general lower bounds. Here $\A$ is an
      arbitrary tropical semiring, $\ddeg{f}{r}$ is the maximum
      possible number of monomials of $f$ containing a fixed monomial
      of degree~$r$, and $r$ is some integer $m/3\leq r\leq 2m/3$.}
    \label{tab:best}
  \end{center}
\end{table}

\section{Conclusion and Open Problems}

In this paper we summarized known and presented some new lower-bound
arguments for tropical circuits, and hence, for the dynamic
programming paradigm; Table~\ref{tab:best} gives a short overview.  We
have also shown that these bounds already yield strong (even
exponential) lower bounds for a full row of important polynomials (see
Table~\ref{tab:bounds}).  Still, the known arguments seem to fail for
non-homogeneous polynomials like $\CONN{}$ or $\STCONN{}$.

Almost exact lower bounds on the \emph{depth} circuits computing these
polynomials are known even in the boolean semiring: $\Theta(\log^2 n)$
for $\STCONN{n}$ proved by \Cite{Karchmer and Wigderson}{KW}, and
$\Omega(\ln^2 n/\ln\ln)$ proved by \Cite{Goldmann and
  H{\aa}stad}{hastad} for $\CONN{n}$; \Cite{Yao}{yao} earlier proved
$\Omega(\ln^{3/2} n/\ln\ln)$ for this latter polynomial. By
Lemma~\ref{lem:bool}, these bounds hold also in tropical semirings.

But the situation with estimating the \emph{size} of circuit for these
polynomial remains unclear. We know (Theorem~\ref{thm:upper-bounds})
that both of them have boolean and $\Min$-circuits of size $O(n^3)$,
but no lower bound larger than a trivial quadratic is known.

\begin{probl}\label{prob:stconn}
  Does $\B(f)=\Omega(n^3)$ or at least $\Min(f)=\Omega(n^3)$ hold for
  $f=\STCONN{n}$ and/or $f=\CONN{n}$?
\end{probl}
Note that the lower bound $\Omega(n^3)$ for the all-pairs shortest
paths polynomial $\APSP{}$, given in Corollary~\ref{cor:APSP} does not
automatically imply the same lower bounds for the connectivity
polynomial $\CONN{}$: a circuit for $\CONN{}$ needs \emph{not} to
compute the polynomials of $\APSP{}$ at \emph{separate} gates.

One could show $\Min(\CONN{})=\Omega(n^3)$ by showing that monotone
arithmetic circuits for the following ``multiplicative version'' of
the triangle polynomial $\tR{n}$ require $\Omega(n^3)$ gates. Recall
that $\tR{n}(x,y,z)=\sum_{i,j\in [n]} z_{ij}\sum_{k\in
  [n]}x_{ik}y_{kj}$.  We already know (see Corollary~\ref{cor:APSP})
that $\R[\tR{n}]=\Theta(n^3)$, and hence also
$\Min(\tR{n})=\Theta(n^3)$ since the polynomial is homogeneous.
Replace now the outer sum by product, and consider the polynomial
$\tRP{n} = \prod_{i,j\in [n]} z_{ij}\sum_{k\in [n]}x_{ik}y_{kj}$.

\begin{probl}
  Does $\R[\tRP{n}]=\Omega(n^3)$?
\end{probl}
If true, this would yield $\Min(\CONN{n})=\Omega(n^3)$, because the
polynomial $\tRP{n}$ is homogeneous (of degree $3n^2$).

\subsubsection*{Acknowledgements}
I am thankful to Georg Schnitger and Igor Sergeev for interesting
discussions.

\small

\end{document}